\documentclass[superscriptaddress,prb,floatfix,aps,amsfonts,amssymb,amsmath,twocolumn]{revtex4-2}

\usepackage{graphics}
\usepackage{amsthm}
\usepackage{tikz-cd}
\usepackage{mathtools}
\usepackage{dcolumn}
\usepackage{color}
\usepackage{hyperref}
\usepackage[ruled,vlined]{algorithm2e}
\usepackage{multirow}
\usepackage{makecell}

\usetikzlibrary{decorations.pathmorphing}

\newtheoremstyle{aps}
  {\topsep}{\topsep}
  {\rmfamily}
  {\parindent}
  {\itshape}{.\,}{ }
  {\thmname{#1}\thmnumber{ #2}\thmnote{ (#3)}}
\theoremstyle{aps}
\newtheorem{thm}{Theorem}
\newtheorem{dfn}{Definition}
\newtheorem{lem}{Lemma}

\SetAlgoCaptionSeparator{.}

\begin{document}
\title{Classification and enumeration of solid-solid phase transition mechanisms}
\author{Fang-Cheng Wang}
\affiliation{State Key Laboratory for Artificial Microstructure and Mesoscopic Physics, Frontier Science Center for Nano-optoelectronics, School of Physics, Peking University, Beijing 100871, People's Republic of China}
\author{Qi-Jun Ye}
\email{qjye@pku.edu.cn}
\affiliation{State Key Laboratory for Artificial Microstructure and Mesoscopic Physics, Frontier Science Center for Nano-optoelectronics, School of Physics, Peking University, Beijing 100871, People's Republic of China}
\affiliation{Interdisciplinary Institute of Light-Element Quantum Materials, Research Center for Light-Element Advanced Materials, and Collaborative Innovation Center of Quantum Matter, Peking University, Beijing 100871, People's Republic of China}
\author{Yu-Cheng Zhu}
\affiliation{State Key Laboratory for Artificial Microstructure and Mesoscopic Physics, Frontier Science Center for Nano-optoelectronics, School of Physics, Peking University, Beijing 100871, People's Republic of China}
\author{Xin-Zheng Li}
\email{xzli@pku.edu.cn}
\affiliation{State Key Laboratory for Artificial Microstructure and Mesoscopic Physics, Frontier Science Center for Nano-optoelectronics, School of Physics, Peking University, Beijing 100871, People's Republic of China}
\affiliation{Interdisciplinary Institute of Light-Element Quantum Materials, Research Center for Light-Element Advanced Materials, and Collaborative Innovation Center of Quantum Matter, Peking University, Beijing 100871, People's Republic of China}
\affiliation{Peking University Yangtze Delta Institute of Optoelectronics, Nantong, Jiangsu 226010, People's Republic of China}

\date{\today}

\begin{abstract}
	Crystal-structure match (CSM), the atom-to-atom correspondence between two crystalline phases, is used extensively to describe solid-solid phase transition (SSPT) mechanisms.
	However, existing computational methods cannot account for all possible CSMs.
	Here, we propose a formalism to classify all CSMs into a tree structure, which is independent of the choices of unit cell and supercell.
	We rigorously proved that only a finite number of noncongruent CSMs are of practical interest.
	By representing CSMs as integer matrices, we introduce the \textsc{crystmatch} method to \textit{exhaustively} enumerate them, which uncontroversially solves the CSM optimization problem under \textit{any} geometric criterion.
	%
	%The power of its early version for orientation relationship screening has already been demonstrated in [Phys. Rev. Lett. 132, 086101 (2024)].
	%
	For most SSPTs, \textsc{crystmatch} can reproduce all known deformation mechanisms and CSMs within 10 CPU minutes, while also revealing thousands of new candidates.
	The resulting database can be further used for comparing experimental phenomena, high-throughput energy barrier calculations, or machine learning.
\end{abstract}

\maketitle

\section{Introduction}\label{sec:intro}
In solid-solid phase transition (SSPT), atoms in one crystal structure are rearranged to form another crystal structure.
This process establishes an atom-to-atom correspondence between the initial and final structures, which is called a crystal-structure match (CSM)~\cite{stevanovic2018predicting,therrien2020matching,wang2024crystal}.
Every transition path on the potential energy surface (PES) determines a CSM, while the latter describes the SSPT more concisely and flexibly.
Consequently, it is common for researchers to identify the CSM, rather than the transition path, to determine which SSPT mechanism is involved~\cite{zahn2005mechanism,shimojo2004atomistic,cai2007first,zahn2004nucleation,badin2021nucleating}.
This is especially true in nucleation simulations~\cite{zahn2004nucleation,badin2021nucleating,santos2022size,khaliullin2011nucleation}, where the term ``mechanism'' is almost synonymous with ``CSM'', which serves as a common language for both nucleation and concerted mechanism studies~\cite{wang2024crystal}.
It has long been believed that SSPTs, like chemical reactions, occur along a minimum energy path (MEP) with the lowest energy barrier~\cite{eyring1935activated,sheppard2012generalized,therrien2021metastable}.
Just as each transition path determines a CSM, each CSM can also yield an MEP through methods such as the solid-state nudged elastic band (NEB)~\cite{sheppard2012generalized}.
However, NEB-like methods can only produce an MEP with the user-specified CSM, ignoring all other CSMs and their MEPs, as illustrated in Fig.~\ref{fig:neb}.
Hence, despite their great success, NEB-like methods may still fail to give the lowest-barrier MEP, which has been an open question in the SSPT mechanism research for the past decade~\cite{therrien2021metastable,zhu2019phase,shang2014stochastic}.
The main difficulty is that unlike the atom-to-atom correspondence in a chemical reaction, possible CSMs in an SSPT are extremely numerous---due to the large number of identical atoms and possible ways of lattice deformation---and thus hard to determine by human intuition~\cite{stevanovic2018predicting,therrien2020matching,wang2024crystal}.
To investigate what CSMs are favored by the SSPT, efforts have been made on two fronts: variable-CSM MEP calculation~\cite{shang2014stochastic,zhu2019phase} and pure geometric CSM optimization~\cite{therrien2021metastable}.
The former includes the stochastic surface walking (SSW)~\cite{shang2014stochastic,guan2015energy,zhang2015variable} and \textsc{pallas}~\cite{zhu2019phase}, where intensive energy calculations hinder the enumeration of CSMs.
The latter relies on a specific optimization criterion, such as minimum strain~\cite{chen2016determination}, minimum number of broken bonds~\cite{stevanovic2018predicting}, or minimum total distance traveled by the atoms~\cite{therrien2020matching,huang2024interface}.
However, geometrically optimized CSMs may not have the lowest energy barrier, and some are even suboptimal under their own criteria~\cite{wang2024crystal}.
It would be useful to list \textit{all} possible CSMs, both to guide MEP calculations and to identify the best CSM under \textit{any} geometric criterion---yet, such a methodology is lacking.
\begin{figure*}[ht!]
  \centering
  \includegraphics[width=0.75\textwidth]{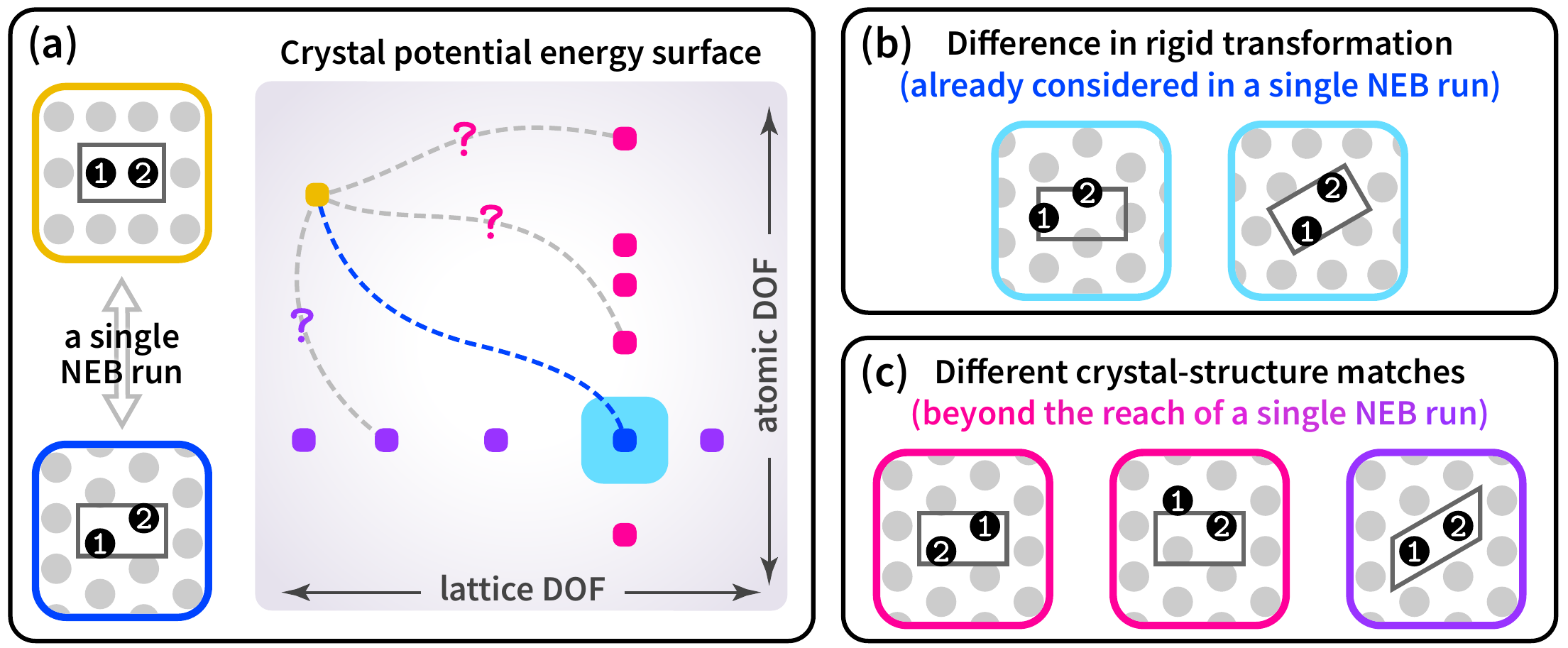}
  \caption{A schematic diagram showing that the initial structure (orthorhombic, yellow) may evolve to the final structure (hexagonal, other colors) in a number of ways.
	(a)
	After specifying the final structure as shown in the blue box, the NEB method can find the MEP denoted by the blue dashed line.
	However, to become the final structure, neither the lattice nor the atoms must end up as the blue box.
	Other possible endpoints spread out in both lattice and atomic degrees of freedom (DOF) on the PES, whose MEPs (gray dashed lines) are inaccessible.
	(b) Endpoints that differ in rigid transformations (cyan) can be easily taken into account by the NEB method~\cite{field1999practical}.
	(c) Endpoints that differ in CSMs, including those with the same SLM (pink) and different SLMs (purple), which we use \textsc{crystmatch} to exhaust.}
  \label{fig:neb}
\end{figure*}
In this paper, we present a complete classification of CSMs, and the \textsc{crystmatch} method for \textit{exhaustively} enumerating them.
To this end, a mathematical formalism is developed in Section~\ref{sec:formal} such that the deformation gradient, sublattice match (SLM), multiplicity and shuffle distance of a CSM can be rigorously defined in a \textit{cell-independent} manner.
These concepts organize all CSMs into a tree-like structure, as shown in Fig.~\ref{fig:tree}, whose matrix representation and pruning scheme are discussed in Section~\ref{sec:rep-prun}.
To traverse this tree, three algorithms are detailed in Section~\ref{sec:algs}.
We demonstrate their application using the B1--B2 and graphite-to-diamond transitions as examples in Section~\ref{sec:app}.
\begin{figure}[b!]
  \centering
  \includegraphics[width=0.95\linewidth]{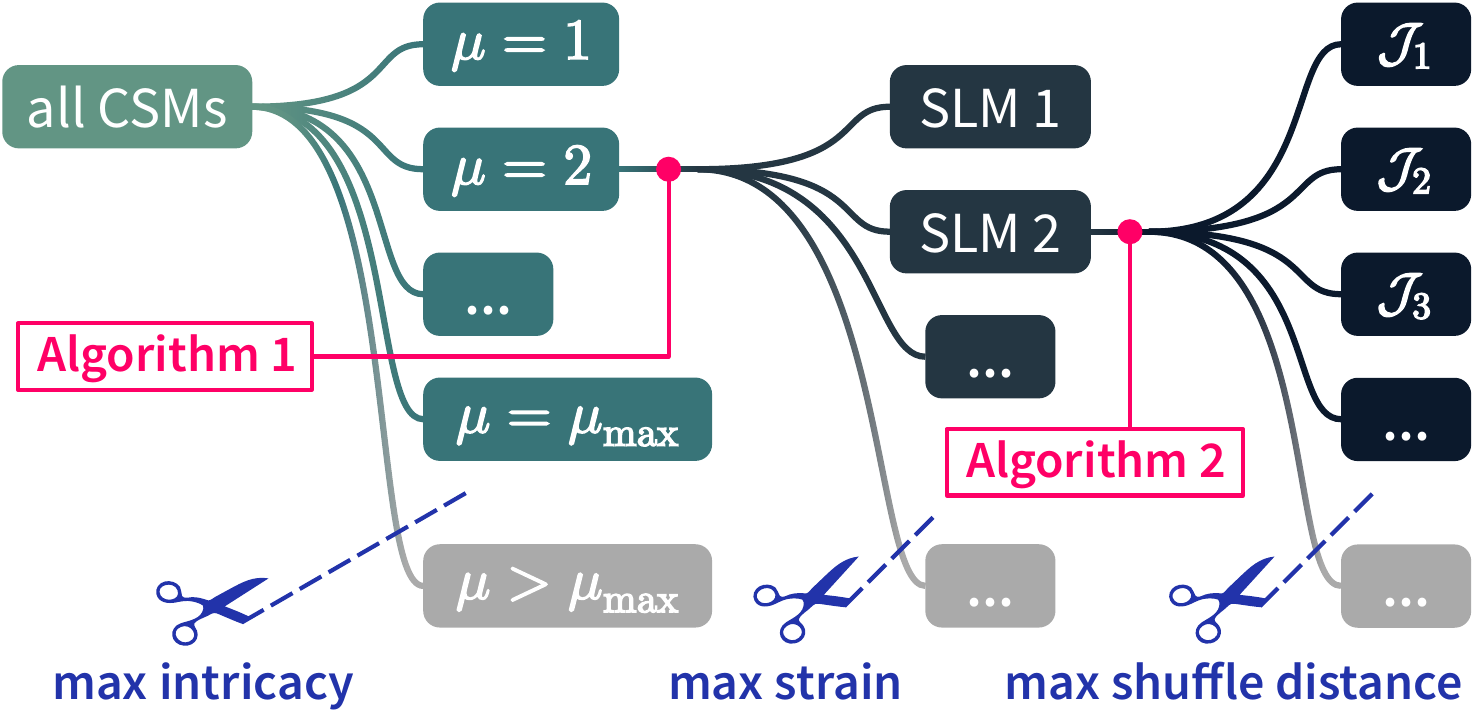}
  \caption{Given the initial and final crystal structures, all CSMs between them can be classified into a tree by their multiplicities and SLMs.
  	Without restrictions, the tree would have infinite breadth at each level.
	For empirical (see Table~\ref{tab:mech}) and computational considerations, an upper bound $\mu_\text{max}$ for the multiplicity $\mu$ is introduced.
	Also, CSMs with too much strain or shuffle distance are deemed nonphysical.
	Under these prunings, the remaining tree contains only a finite number of CSMs, which can be obtained via Algorithms~\ref{alg:imt} and \ref{alg:pct}.}
  \label{fig:tree}
\end{figure}
It should be noted that the molecular dynamics (MD) has also been used extensively to study SSPTs~\cite{scandolo1995pressure}.
Among different transition paths, the one with the lowest energy barrier is most likely to be obtained by MD.
This bypasses the CSM problem, but it takes an average time $\bar t\propto\exp(\beta \Delta^\ddagger )$ to see a rare event with energy barrier $\Delta^\ddagger $ to happen.
Such $\Delta^\ddagger$ increases linearly with the supercell volume until it can accommodate a critical nucleus~\cite{sheppard2012generalized,khaliullin2011nucleation,badin2021nucleating,santos2022size}, making the MD simulation very hard---if not infeasible.
Fortunately, enhanced sampling techniques like the metadynamics (MetaD)~\cite{laio2002escaping,martovnak2003predicting} can reinforce the occurrence of rare events by adding a history-dependent bias potential, making it possible to simulate nucleation in SSPTs~\cite{badin2021nucleating,santos2022size}.
Nevertheless, the bias potential may elevate different energy barriers \textit{unequally}, thus altering the lowest-barrier path and CSM.
We do not discuss here how \textsc{crystmatch} can serve MetaD, but simply show that the enumeration results of the former contain the CSM yielded by the latter.
\begin{figure}[b!]
  \centering
  \includegraphics[width=\linewidth]{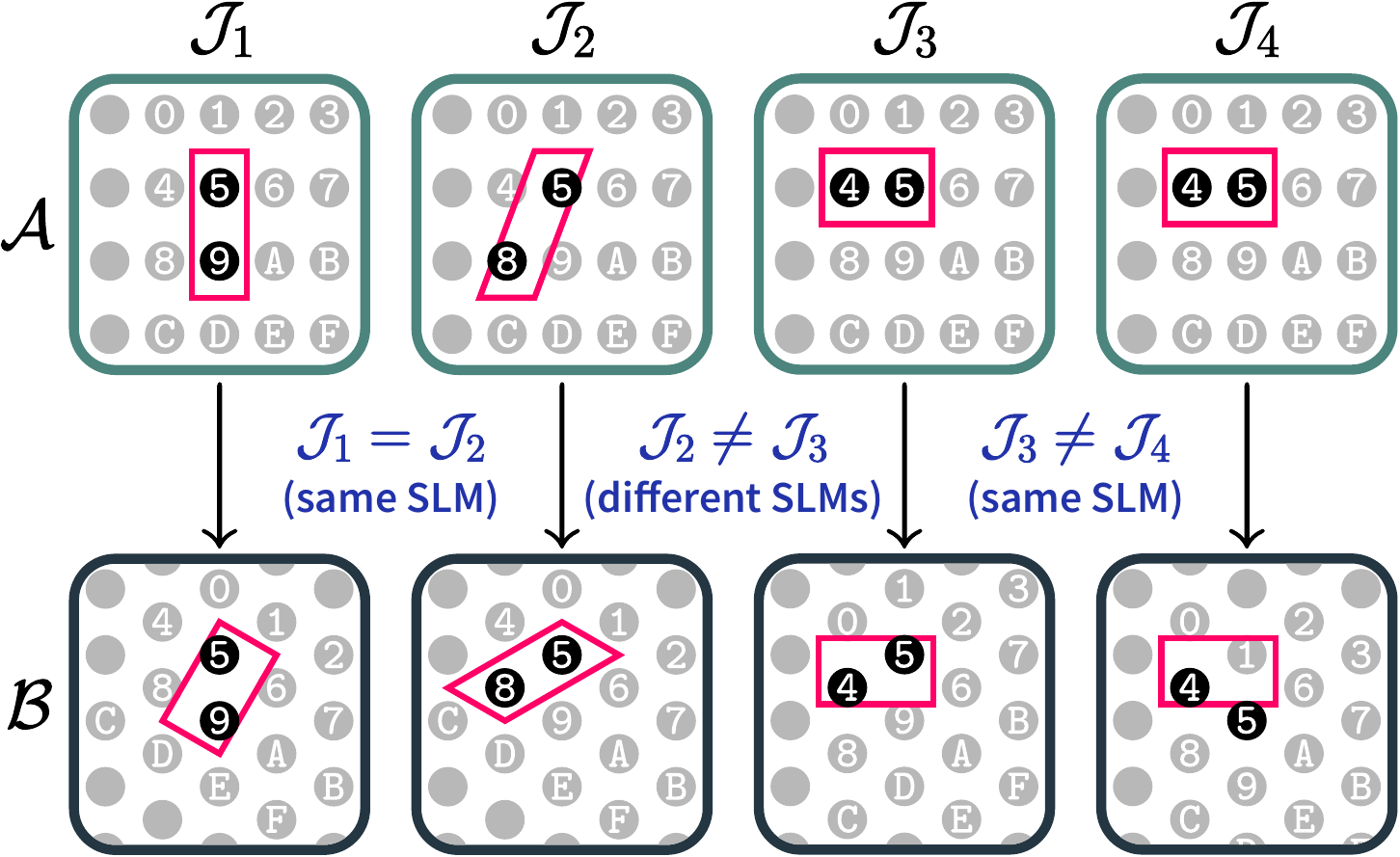}
  \caption{Four CSMs from $\mathcal{A}$ (orthorhombic) to $\mathcal{B}$ (hexagonal).
	Each CSM is described by the correspondences between the initial and final supercell vectors (pink), and two translational-inequivalent atoms (black).
	From the labels on the atoms (\texttt{0-9} and \texttt{A-F}) in $\mathcal{B}$, one can see that $\mathcal{J}_1=\mathcal{J}_2\neq \mathcal{J}_3\neq \mathcal{J}_4$.}
  \label{fig:diff}
\end{figure}
\section{Cell-independent formalism}\label{sec:formal}
Conventionally, CSMs are described by a pair of initial and final supercells, as shown in Fig.~\ref{fig:diff}.
However, the multivaluedness of this description causes severe inconveniences for both the theory and the enumeration of CSMs~\cite{wang2024crystal,therrien2020matching,li2022smallest}.
The purpose of this section is to establish a \textit{cell-independent} formalism, making it self-evident that multiplicity and SLM---the main subjects in Fig.~\ref{fig:tree}---are general, singlevalued, intrinsic properties of CSMs.
We also define congruence relations among CSMs and SLMs, which are crucial for the efficiency of \textsc{crystmatch}.
To ensure the unambiguousness of the narration and proofs, we define fundamental terms (e.g., ``lattice'' and ``crystal structure'') using the language of set theory.
All lemmas are proved in Appendix~\ref{append:lemma}, while frequently used symbols are summarized in Table~\ref{tab:symbol-formal}.
Commutative diagrams, as a tool for depicting the composition of homomorphisms, are also adopted to simplify the exposition.
Let us begin with a set of points in $\mathbb{R}^3$, where each point is endowed with an atomic species; we denote by $\mathbb{X}$ the set of all atomic species.
\begin{dfn}[Atomic Structure]\label{def:as}
	An atomic structure is a nonempty set $\mathcal{A}\subset \mathbb{R}^3$ together with a mapping $\chi_A\colon \mathcal{A}\to\mathbb{X}$.
\end{dfn}
\begin{dfn}[Atom-to-Atom Correspondence]\label{def:ata}
    Let $\mathcal{A}$ and $\mathcal{B}$ be atomic structures.
	We say that a mapping $\mathcal{J}\colon\mathcal{A}\to\mathcal{B}$ is an atom-to-atom correspondence if
	(1) $\mathcal{J}$ is bijective, and
	(2) $\mathcal{J}$ preserves atom species, i.e., makes the following diagram commute:
\begin{equation}\label{ajbta}
	\begin{tikzcd}[column sep=small]
		\mathcal{A}\arrow[rr, "{\mathcal{J}}"]&&\mathcal{B}\\
						      &\mathbb{X}\arrow[from=ul, "\chi_A"']\arrow[from=ur, "\chi_B"]&
  \end{tikzcd}
.\end{equation}
\end{dfn}
The above diagram is essentially a directed graph, whose vertices are sets and arrows are mappings.
Paths composed of multiple arrows naturally represent a composite mapping, from its starting point to its endpoint.
Commutativity of Eq.~(\ref{ajbta}) means that any two paths with the same starting point and endpoint are equal, i.e.,
\begin{equation}\label{tatbj}
		\chi_A=\chi_B\circ\mathcal{J}
.\end{equation}
One can also interpret Eq.~(\ref{ajbta}) as follows: For any $\mathbf{a}\in\mathcal{A}$, it is mapped along different paths to
\begin{equation}
	\chi_A(\mathbf{a})=\chi_B(\mathcal{J}(\mathbf{a}))
,\end{equation}
which means that the species of atom $\mathbf{a}$ is the same as that of its counterpart $\mathcal{J}(\mathbf{a})$.
Commutative diagrams are primarily used in category theory and homological algebra, where the arrows represent homomorphisms.
In this spirit, we redefine for all subsequent diagrams:
Arrows between two atomic structures \textit{always} represent atom-to-atom correspondences, i.e., species-preserving bijections.
Consequently, the symbols $\chi_A$ and $\mathbb{X}$ no longer appear in most diagrams.
This approach handles polyatomic SSPTs once and for all.
We hope the reader keeps in mind that $\!\!\begin{tikzcd}\mathcal{A}\arrow[r,"\ldots"]&\mathcal{B}\end{tikzcd}\!\!$ denotes an ``isomorphism'' between atomic structures.
\begin{table}[t!]
	\caption{Frequently used symbols in Section~\ref{sec:formal}.}\label{tab:symbol-formal}
\begin{ruledtabular}
	\begin{tabular}{cl}
		Symbol & Meaning \\ \hline
		$\mathcal{A},\mathcal{B}$ & Crystal structure (atomic structure) \\
	$S$ & Nonsingular linear transformation \\
	$+\mathbf{t}$ & Translation by $\mathbf{t}\in\mathbb{R}^3$\\
	$S\mathcal{A}+\mathbf{t}$ & Structure deformed by $S$, then translated by $\mathbf{t}$ \\
	$L_A$ & Lattice of a crystal structure $\mathcal{A}$ \\
	$\tilde L_A$ & Sublattice of a crystal structure $\mathcal{A}$ \\
	$Z_A$ & Number of atoms in a primitive cell of $\mathcal{A}$ \\
	$\mathcal{J}$ & CSM (atom-to-atom correspondence) \\
	$(\tilde L_A,\tilde L_B,S)$ & An SLM from $\mathcal{A}$ to $\mathcal{B}$ \\
	$\tilde Z$ & Period of an SLM or a CSM \\
	$\mu$ & Multiplicity of an SLM or a CSM\\
	$(S,\mathbf{t})$ & Affine transformation $\mathbf{a}\mapsto S\mathbf{a}+\mathbf{t}$
\end{tabular}
\end{ruledtabular}
\end{table}
\subsection{Lattices and crystal structures}
\begin{dfn}[Lattice]\label{def:l}
	We say that a vector set $L\subset \mathbb{R}^3$ is a lattice if it (1) is a group under vector addition, (2) is full rank (has three linealy independent elements), and (3) has a positive lower bound $\lambda$ on the distance between any two of its elements, i.e.,
	\begin{equation}\label{t1t2l}
	  \forall \mathbf{t}_1,\mathbf{t}_2\in L,\quad \mathbf{t}_1\neq\mathbf{t}_2\implies\left|\mathbf{t}_1-\mathbf{t}_2\right|\ge \lambda
	.\end{equation}
\end{dfn}
\begin{dfn}[Crystal Structure]\label{def:la}
	Let $\mathcal{A}$ be an atomic structure.
	We say that a vector $\mathbf{t}\in\mathbb{R}^3$ is a translation element of $\mathcal{A}$ if the mapping $\mathbf{a}\mapsto \mathbf{a}+\mathbf{t}$ is an atom-to-atom correspondence from $\mathcal{A}$ to itself, i.e.,
\begin{equation}
	\begin{tikzcd}
		\mathcal{A}\arrow[r,"+\mathbf{t}"]&\mathcal{A}
	\end{tikzcd}
.\end{equation}
	We say that an atomic structure $\mathcal{A}$ is a crystal structure if all its translation elements form a lattice, denoted as $L_A$.
	We refer to $L_A$ as the lattice of $\mathcal{A}$, and sublattices of $L_A$ sublattices of $\mathcal{A}$.
\end{dfn}
Let $\tilde L_A$ be a sublattice of a crystal structure $\mathcal{A}$.
The relation $\sim$ on $\mathcal{A}$ defined as
\begin{equation}\label{a1a2t}
	\mathbf{a}\sim\mathbf{a}'\quad\iff\quad \exists  \mathbf{t}\in\tilde L_A,\quad\mathbf{a}'=\mathbf{a}+\mathbf{t}
\end{equation}
is an equivalence relation, whose reflexivity, symmetry, and transitivity are ensured by the identity element, invertibility, and closure of the group $(\tilde L_A,+)$, respectively.
We say that $\mathbf{a},\mathbf{a}'\in\mathcal{A}$ are $\tilde L_A$-equivalent, and denote the $\tilde L_A$-equivalence class of $\mathbf{a}$ by
\begin{equation}
  \mathbf{a}+\tilde L_A=\left\{\mathbf{a}+\mathbf{t}\,\middle|\,\mathbf{t}\in\tilde L_A\right\}
,\end{equation}
and the quotient set by
\begin{equation}
  \mathcal{A}/\tilde L_A=\left\{\mathbf{a}+\tilde L_A\,\middle|\,\mathbf{a}\in\mathcal{A}\right\}
.\end{equation}
For any crystal structure $\mathcal{A}$, we denote $Z_A=\lvert\mathcal{A}/L_A\rvert$, which equals to the number of atoms in a primitive cell.
Actual crystal structures all have a finite $Z_A$ since they have a finite atomic density.
We will assume this henceforth, but the reader can also modify Definition~\ref{def:as} to make it rigorous; see Lemma~\ref{lem:finite-z}.
\begin{dfn}[CSM]
	A crystal structure match (CSM) is an atom-to-atom correspondence between two crystal structures.
	We denote the set of all CSMs from crystal structure $\mathcal{A}$ to crystal structure $\mathcal{B}$ by $\operatorname{CSM}(\mathcal{A},\mathcal{B})$.
\end{dfn}
\subsection{Decomposing crystal-structure matches}
Now we aim at defining the deformation gradient of a CSM.
Our approach is to first define the \textit{shuffle}, a type of CSM that is considered to have no deformation.
Then, we will define the deformation of a general CSM as the difference between that CSM and a shuffle, as shown in Fig.~\ref{fig:slm}(a).
In the following text, the composition of mappings (including CSMs, linear transformations, and translations) will appear more frequently.
When there is no risk of confusion, we will denote $f \circ g$ simply as $fg$.
\begin{dfn}[Shuffle]\label{def:shuffle}
	Let $\mathcal{A}$ and $\mathcal{B}$ be crystal structures.
	We say that a vector $\mathbf{t}\in L_A\cap L_B$ is a translation element of $\mathcal{J}\in\operatorname{CSM}(\mathcal{A},\mathcal{B})$ if we have the commutative diagram:
	\begin{equation}\label{ajbtt}
	  \begin{tikzcd}
		  \mathcal{A}\arrow[r, "{\mathcal{J}}"]&\mathcal{B}\\
		  \mathcal{A}\arrow[from=u, "{+\mathbf{t}}"']\arrow[r, "{\mathcal{J}}"']&\mathcal{B}\arrow[from=u, "{+\mathbf{t}}"]
	  \end{tikzcd}
	,\end{equation}
	i.e., $\mathcal{J}(\mathbf{a}+\mathbf{t})=\mathcal{J}(\mathbf{a})+\mathbf{t}$ for all $\mathbf{a}\in\mathcal{A}$.
	If all translation elements of $\mathcal{J}$ form a lattice, we refer to it as the shuffle lattice of $\mathcal{J}$ and say that $\mathcal{J}$ is a shuffle.
	We use a squiggly arrow $\!\!\begin{tikzcd}\mathcal{A}\arrow[r,rightsquigarrow,"\cdots"]&\mathcal{B}\end{tikzcd}\!\!$ to denote a shuffle.
\end{dfn}
In actual SSPTs, $\mathcal{A}$ and $\mathcal{B}$ in different phases generally have no common translation elements (except $\mathbf{0}$), in which case $\mathcal{J}\in\operatorname{CSM}(\mathcal{A},\mathcal{B})$ can never be a shuffle.
However, it is possible to deform $\mathcal{A}$ such that the deformed structure shares more translation elements with $\mathcal{B}$, making the deformed CSM a shuffle.
To formalize this idea, we denote by $ S\mathcal{A}=\left\{S\mathbf{a}\,\middle|\,\mathbf{a}\in\mathcal{A}\right\} $ the crystal structure deformed from $\mathcal{A}$ via $S\in\operatorname{GL}(3,\mathbb{R})$.
It has a canonical atom-species mapping $\chi_AS^{-1}$, which is the only mapping that makes the following diagram commutes:
\begin{equation}\label{assac}
  \begin{tikzcd}
	\mathcal{A}\arrow[r, "S"]&S\mathcal{A}\\
	\mathbb{X}\arrow[from=u, "\chi_A"']\arrow[from=ur, "\chi_AS^{-1}"]&
  \end{tikzcd}
,\end{equation}
thus making $S\in\operatorname{CSM}(\mathcal{A},S\mathcal{A})$.
Now we draw the commutative diagram corresponding to Fig.~\ref{fig:slm}(a).
\begin{figure}[t!]
  \centering
  \includegraphics[width=\linewidth]{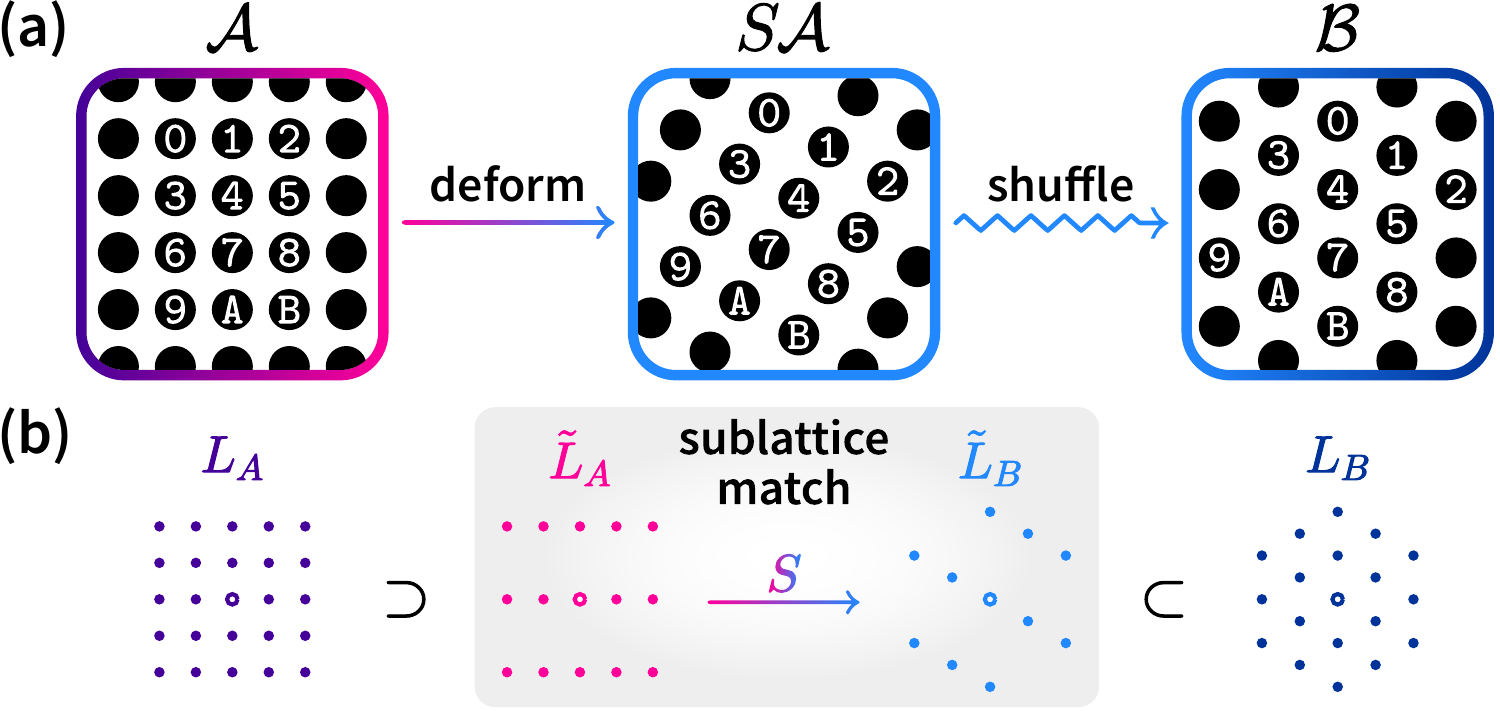}
  \caption{(a) An example of $\mathcal{J}\in\operatorname{CSM}(\mathcal{A},\mathcal{B})$ being decomposed into a deformation gradient $S\in\operatorname{CSM}(\mathcal{A},S\mathcal{A})$ and a shuffle $\mathcal{J}S^{-1}\in\operatorname{CSM}(S\mathcal{A},\mathcal{B})$. (b) In the above decomposition, a sublattice $\tilde L_A\subset L_A$ is deformed to the shuffle lattice $\tilde L_B\subset L_B$. We call the triplet $(\tilde L_A,\tilde L_B,S)$ the SLM of $\mathcal{J}$.}
  \label{fig:slm}
\end{figure}
\begin{dfn}[Deformation Gradient]\label{def:dg}
	We say that $S\in\operatorname{GL}(3,\mathbb{R})$ is a deformation gradient of $\mathcal{J}\in\operatorname{CSM}(\mathcal{A},\mathcal{B})$ if we have the commutative diagram:
	\begin{equation}\label{ajbss}
		\begin{tikzcd}[column sep=small]
			\mathcal{A}\arrow[rr, "{\mathcal{J}}"]\arrow[rd, "S"']&&\mathcal{B}\\
			&S\mathcal{A}\arrow[ur, rightsquigarrow]&
	  \end{tikzcd}
	,\end{equation}
	i.e., $\mathcal{J}S^{-1}\colon S\mathcal{A}\to\mathcal{B}$ is a shuffle.
\end{dfn}
\subsection{Classifying crystal-structure matches}
Next, we will show the uniqueness of the decomposition in Fig.~\ref{fig:slm}(a).
This makes the concept of SLM well-defined, which is crucial for classifying and representing CSMs, as shown in Fig.~\ref{fig:slm}(b).
\begin{thm}\label{thm:s}
	The deformation gradient of a CSM is either unique or nonexistent.
\end{thm}
\begin{proof}[\indent Proof]
	Assume that $\mathcal{J}\in\operatorname{CSM}(\mathcal{A},\mathcal{B})$ has two deformation gradients $S_1$ and $S_2$.
	By Definition~\ref{def:dg} we have the commutative diagram:
	\begin{equation}\label{as2s2}
		\begin{tikzcd}[column sep=large]
		  \mathcal{A}\arrow[r,"S_2"]\arrow[d,"S_1"']\arrow[rd,"\mathcal{J}"']&S_2\mathcal{A}\\
		  S_1\mathcal{A}&\mathcal{B}\arrow[from=u,rightsquigarrow,"\mathcal{J}S_2^{-1}"]\arrow[from=l,rightsquigarrow,"\mathcal{J}S_1^{-1}"']
	  \end{tikzcd}
	.\end{equation}
	We denote the shuffle lattices of $\mathcal{J}S_1^{-1}$ and $\mathcal{J}S_2^{-1}$ by $\tilde L_1$ and $\tilde L_2$, respectively.
	For each $\mathbf{t}\in\tilde L_1\cap\tilde L_2$, we have the commutative diagram:
	\begin{equation}\label{s1ajs}
		\begin{tikzcd}[column sep=large]
		  S_1\mathcal{A}\arrow[r, "\mathcal{J}S_1^{-1}",rightsquigarrow]&\mathcal{B}&S_2\mathcal{A}\arrow[l, "\mathcal{J}S_2^{-1}"',rightsquigarrow]\\
		  S_1\mathcal{A}\arrow[r, "\mathcal{J}S_1^{-1}"',rightsquigarrow]&\mathcal{B}&S_2\mathcal{A}\arrow[l, "\mathcal{J}S_2^{-1}",rightsquigarrow]
		  \arrow[from=1-1, to=2-1, "+\mathbf{t}"']
		  \arrow[from=1-2, to=2-2, "+\mathbf{t}"']
		  \arrow[from=1-3, to=2-3, "+\mathbf{t}"]
	  \end{tikzcd}
	.\end{equation}
	Combining Eqs.~(\ref{as2s2}) and (\ref{s1ajs}), we obtain:
	\begin{equation}\label{cube}
	\begin{tikzcd}[row sep=scriptsize, column sep=scriptsize]
& \mathcal{A} \arrow[rd, "\mathcal{J}"', near start]\arrow[dl, "S_1"', near start] \arrow[rr, "S_2", near start] & & S_2\mathcal{A} \arrow[dl, "\mathcal{J}S_2^{-1}", rightsquigarrow, near end] \arrow[dd, "+\mathbf{t}", near start] \\ S_1\mathcal{A} \arrow[rr, crossing over, "\mathcal{J}S_1^{-1}"', rightsquigarrow] \arrow[dd, "+\mathbf{t}"', near start] & & \mathcal{B} \\
& \mathcal{A} \arrow[rd, "\mathcal{J}"', near start]\arrow[dl, "S_1"', near start] \arrow[rr, "S_2", near start] & & S_2\mathcal{A} \arrow[dl, "\mathcal{J}S_2^{-1}", rightsquigarrow, near end] \\
		S_1\mathcal{A} \arrow[rr, crossing over, "\mathcal{J}S_1^{-1}"', rightsquigarrow] & & \mathcal{B} \arrow[from=uu, crossing over, "+\mathbf{t}", near start]
\end{tikzcd}
	,\end{equation}
	from which we can extract the subdiagram:
	\begin{equation}\label{as1s1}
		\begin{tikzcd}
			S_1\mathcal{A}&\mathcal{A}\arrow[l,"S_1"']\arrow[r,"S_2"]&S_2\mathcal{A}\\
			S_1\mathcal{A}&\mathcal{A}\arrow[l,"S_1"']\arrow[r,"S_2"]&S_2\mathcal{A}
		  \arrow[from=1-1, to=2-1, "+\mathbf{t}"']
		  \arrow[from=1-3, to=2-3, "+\mathbf{t}"]
		\end{tikzcd}
	.\end{equation}
	The commutativity of Eq.~(\ref{as1s1}) means that $S_1^{-1}\mathbf{t}=S_2^{-1}\mathbf{t}$.
	By Definitions~\ref{def:shuffle}, both $\tilde L_1$ and $\tilde L_2$ are sublattices of $L_B$ so that $\tilde L_1\cap \tilde L_2$ is full rank (Lemma~\ref{lem:intersect}).
	Hence, there exist three linearly independent vectors $\mathbf{t}_1,\mathbf{t}_2,\mathbf{t}_3$ such that
	\begin{equation}
		S_1^{-1}[\mathbf{t}_1,\mathbf{t}_2,\mathbf{t}_3]=S_2^{-1}[\mathbf{t}_1,\mathbf{t}_2,\mathbf{t}_3]
	.\end{equation}
	Since $[\mathbf{t}_1,\mathbf{t}_2,\mathbf{t}_3]$ is invertible, we have $S_1=S_2$.
\end{proof}
Theorem~\ref{thm:s} enables us to classify CSMs by their deformation and shuffle.
Now, we merge the deformation gradient and the shuffle lattice into a single concept, whose utility will become evident in Sections~\ref{sec:rep-prun} and \ref{sec:algs}.
\begin{table*}[t!]
	\caption{Some previously proposed CSMs (also known as mechanisms or paths).
	Congruent CSMs (e.g., the Bain mechanism and many shear mechanisms of the FCC-to-BCC transition~\cite{therrien2020minimization}) are only displayed once.
	See Ref.~\citenum{capillas2007maximal} for more details.}\label{tab:mech}
\begin{ruledtabular}
\begin{tabular}{clrrrlcc}
	SSPT prototype&
	CSM\footnote{The abbreviations stand for Wang-Ye-Zhu-Li (WYZL), Watanabe-Tokonami-Morimoto (WTM), Tol{\'e}dano-Knorr-Ehm-Depmeier (TKED), Therrien-Graf-Stevanovi{\'c} (TGS), Tolbert-Alivisatos-Sowa (TAS), and Zhu-Cohen-Strobel (ZCS).} &
$\mu$ & $\tilde Z$ &
RMSS\footnote{The CSMs in the table are defined independently of specific materials. However, to calculate the root-mean-square strain (RMSS) and root-mean-square displacement (RMSD), lattice parameters must be specified. Here, lattice parameters of Fe, CsCl, and ZnO are used.\label{tabnote1}} &
RMSD\textsuperscript{\ref{tabnote1}} &
Discovered via &
Experiment/simulation evidence\footnote{The orientation relationship (OR) of a CSM is defined in the rotation-free manner, which is also adopted by Refs.~\citenum{wang2024crystal,shimojo2004atomistic,stokes2004mechanisms,therrien2020minimization}.} \\
\hline
\multirow{4}{*}{\makecell{A1--A2\\(FCC-to-BCC)}} & Bain & 1 & 1 & 15.9\% & 0 & Inference~\cite{bain1924nature} &
MD~\cite{sandoval2009bain} \\
					     & Therrien-Stevanovi{\'c} & 6 & 6 & 9.0\% & 0.713\,\AA & \textsc{p2ptrans}~\cite{therrien2020minimization} & Pitsch OR~\cite{pitsch1959martensite} \\
					     & WYZL $\mathcal{J}_2$ & 6 & 6 & 9.0\% & 0.884\,\AA & \textsc{crystmatch}~\cite{wang2024crystal} & Nishiyama-Wassermann OR~\cite{nishiyama2012martensitic} \\
					     & WYZL $\mathcal{J}_1$ & 36 & 36 & 4.3\% & 0.961\,\AA & \textsc{crystmatch}~\cite{wang2024crystal} & Kurdjumov-Sachs OR~\cite{kurdjumow1930mechanismus} \\
\hline &\\[-2.5ex]
\multirow{4}{*}{\makecell{B1--B2\\(NaCl-to-CsCl)}} & Buerger & 1 & 2 & 27.5\% & 0 & Inference~\cite{buerger1951phase}
						   & MD~\cite{nga1992mechanism,zahn2004molecular}
 \\
						   & WTM & 2 & 4 & 17.4\% & 1.135\,\AA &
Inference~\cite{watanabe1977transition} & Watanabe-Blaschko OR~\cite{watanabe1977transition,blaschko1979investigations}
 \\
					& TKED & 4 & 8 & 17.4\% & 1.605\,\AA & Inference~\cite{toledano2003phenomenological}
					& MD~\cite{zahn2004nucleation}, MetaD~\cite{badin2021nucleating}
 \\
							& TGS B1--B2 & 6 & 12 & 10.9\% & 1.654\,\AA &
\textsc{p2ptrans}~\cite{therrien2020matching} & Very low strain \\
\hline &\\[-2.5ex]
\multirow{4}{*}{\makecell{B1--B4\\(NaCl-to-wurtzite)}} & TAS & 1 & 4 & 27.1\% & 0.500\,\AA & Inference~\cite{tolbert1995wurtzite,sowa2001transition}
					   &  MD~\cite{wilson2002transformations,zahn2005mechanism}, MetaD~\cite{santos2022size}
 \\
					   & ZCS II & 1 & 4 & 18.4\% & 0.822\,\AA & \textsc{pallas}~\cite{zhu2019phase}
					   & See Ref.~\citenum{zhu2019phase} \\
				    & Shimojo \textit{et al.} II
				    & 2 & 8 & 15.0\% & 0.865\,\AA & MD~\cite{shimojo2004atomistic}
				    & MD~\cite{shimojo2004atomistic,cai2007first}
 \\
				    & TGS B1--B4 & 3 & 12 & 8.5\% & 1.084\,\AA &
\textsc{p2ptrans}~\cite{therrien2020matching} & Very low strain \\[-2pt]
\end{tabular}
\end{ruledtabular}
\end{table*}
\begin{dfn}[SLM]
	Let $\mathcal{A}$, $\mathcal{B}$ be crystal structures, $\tilde L_A$, $\tilde L_B$ their respective sublattices, and $S\in\operatorname{GL}(3,\mathbb{R})$.
	We say that $(\tilde L_A,\tilde L_B,S)$ is a sublattice match (SLM) from $\mathcal{A}$ to $\mathcal{B}$ if
	\begin{align}\label{alabl}
		\lvert \mathcal{A}/\tilde L_A\rvert &=\lvert \mathcal{B}/\tilde L_B\rvert,\\
		SL_A&=L_B\label{slalb}
	.\end{align}
	We refer to the value of Eq.~(\ref{alabl}) as the period of the SLM, denoted by $\tilde Z$, and say that the SLM has multiplicity
	\begin{equation}\label{mzlcm}
		\mu=\frac{\tilde Z}{\operatorname{lcm}(Z_A,Z_B)}
	.\end{equation}
	Denote by $\operatorname{SLM}(\mathcal{A},\mathcal{B})$ the set of all SLMs from $\mathcal{A}$ to $\mathcal{B}$, and those with multiplicity $\mu$ by $\operatorname{SLM}(\mathcal{A},\mathcal{B};\mu)$.
\end{dfn}
\begin{thm}\label{thm:slm}
	If $\mathcal{J}\in\operatorname{CSM}(\mathcal{A},\mathcal{B})$ has a deformation gradient $S$, then $(\tilde L_A,\tilde L_B,S)\in\operatorname{SLM}(\mathcal{A},\mathcal{B})$, where $\tilde L_B$ is the shuffle lattice of $\mathcal{J}S^{-1}$ and $\tilde L_A=S^{-1}\tilde L_B$.
\end{thm}
\begin{proof}[\indent Proof]
	Eq.~(\ref{slalb}) is already satisfied and we only need to show Eq.~(\ref{alabl}).
	Consider a mapping $p\colon \mathcal{A}/\tilde L_A\to\mathcal{B}/\tilde L_B$ such that the following diagram commutes:
	\begin{equation}\label{ajbla}
	\begin{tikzcd}[column sep=scriptsize]
		\mathcal{A}\arrow[d, "+\tilde L_A"']\arrow[r, "\mathcal{J}"]&\mathcal{B}\arrow[d, "+\tilde L_B"]\\
		\mathcal{A}/\tilde L_A\arrow[r, "p"']&\mathcal{B}/\tilde L_B
	\end{tikzcd}
	,\end{equation}
	i.e., $p\colon\mathbf{a} + \tilde{L}_A \mapsto \mathcal{J}(\mathbf{a}) + \tilde{L}_B$.
	To show that $p$ is well-defined, we must prove its singlevaluedness.
	For any $\mathbf{a}_1$ and $\mathbf{a}_2$ in $\mathcal{A}$, we have
	\begin{align}
		\phantom{\iff{}}&\quad\mathbf{a}_1+\tilde L_A=\mathbf{a}_2+\tilde L_A\label{a1laa}\\
	    \iff{}&\quad \exists \mathbf{t}\in\tilde L_A,\quad \mathbf{a}_1=\mathbf{a}_2+\mathbf{t}\\
	    \iff{}&\quad\exists \mathbf{t}\in\tilde L_A,\quad S\mathbf{a}_1=S\mathbf{a}_2+S\mathbf{t}\\
	    \iff{}&\quad \exists \mathbf{t}'\in\tilde L_B,\quad S\mathbf{a}_1=S\mathbf{a}_2+\mathbf{t}'\\
	    \iff{}&\quad \exists \mathbf{t}'\in\tilde L_B,\quad \mathcal{J}(\mathbf{a}_1)=\mathcal{J}(\mathbf{a}_2)+\mathbf{t}'\label{tlba1}\\
	    \iff{}&\quad \mathcal{J}(\mathbf{a}_1)+\tilde L_B=\mathcal{J}(\mathbf{a}_2)+\tilde L_B\label{ja1lb}
	,\end{align}
	where the shuffle property of $\mathcal{J}S^{-1}\colon S\mathcal{A}\rightsquigarrow\mathcal{B}$ is used to derive Eq.~(\ref{tlba1}).
	The equivalence of Eqs.~(\ref{a1laa}) and (\ref{ja1lb}) means that $p$ is injective, while $p$ is surjective since $\mathcal{J}$ is surjective.
	So far, we can see that $p$ is bijective and thus Eq.~(\ref{alabl}) holds.
\end{proof}
\begin{figure}[b!]
  \centering
  \includegraphics[width=\linewidth]{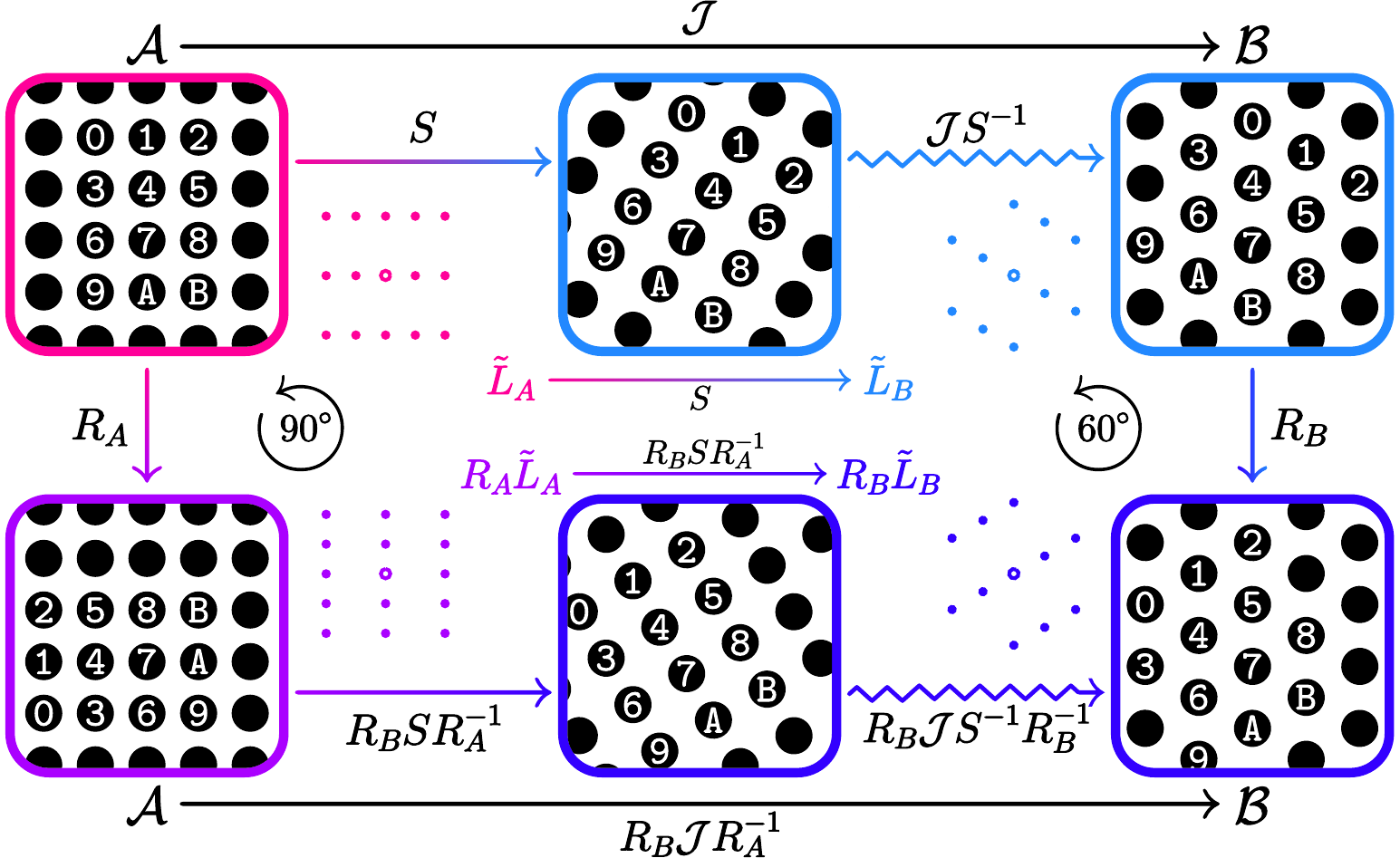}
  \caption{An example of $R_A$ (a rotation by $90^\circ$) and $R_B$ (a rotation by $60^\circ$) being symmetry operations of $\mathcal{A}$ (tetragonal) and $\mathcal{B}$ (hexagonal), respectively.
	  For each $\mathcal{J}\in\operatorname{CSM}(\mathcal{A},\mathcal{B})$, there exists a congruent $R_B\mathcal{J}R_A^{-1}\in\operatorname{CSM}(\mathcal{A},\mathcal{B})$.
	Their respective SLMs are $(\tilde L_A,\tilde L_B,S)$ and $(R_A\tilde L_A,R_B\tilde L_B,R_BSR_A^{-1})$.}
	\label{fig:cong}
\end{figure}
\begin{dfn}[Multiplicity]\label{def:mu}
	If $\mathcal{J}\in\operatorname{CSM}(\mathcal{A},\mathcal{B})$ has a deformation gradient, we say that the $(\tilde L_A,\tilde L_B,S)$ defined in Theorem~\ref{thm:slm} is the SLM of $\mathcal{J}$.
	In such case, we refer to the period and multiplicity of the SLM as those of $\mathcal{J}$.
	If the deformation gradient of $\mathcal{J}$ does not exist, we say that the period and multiplicity of $\mathcal{J}$ are $\infty$.
\end{dfn}
The multiplicity $\mu$ defined in Eq.~(\ref{mzlcm}) is always a positive integer, since the period $\tilde Z$ is a common multiple of $Z_A$ and $Z_B$ (Lemma~\ref{lem:quotient}).
The value of $\mu$ and $\tilde Z$ reflect the intricacy of the CSM---the size of the \textit{smallest} supercell required to describe it.
It must be emphasized that $\mu$ and $\tilde Z$ are properties of the CSM, which are conceptually distinct from the number of atoms used in SSPT simulations.
For example, MetaD calculations using up to ${\sim}10^5$ atoms often yield CSMs with $\mu$ and $\tilde Z$ less than ${\sim}10^1$; see Table~\ref{tab:mech}.
Practically, we focus on solely those CSMs with finite multiplicity.
\subsection{Congruence relations}
Applying rigid transformations to $\mathcal{A}$ and $\mathcal{B}$ yields ``congruent'' yet set-theoretically different CSMs, as illustrated in Fig.~\ref{fig:cong}.
Such CSMs are equivalent as inputs to NEB-like methods, so we only need to enumerate all ``noncongruent'' CSMs.
We denote by $ S\mathcal{A}+\mathbf{t}=\{S\mathbf{a}+\mathbf{t}\,|\,\mathbf{a}\in\mathcal{A}\} $ the crystal structure transformed from $\mathcal{A}$ via an affine transformation $(S,\mathbf{t})$.
In the same sense as in Eq.~(\ref{assac}), we have $(S,\mathbf{t})\in\operatorname{CSM}(\mathcal{A},S\mathcal{A}+\mathbf{t})$.
\begin{dfn}[Congruent CSM]\label{def:cong}
	We say that a CSM $\mathcal{J}'$ is congruent to $\mathcal{J}\in\operatorname{CSM}(\mathcal{A},\mathcal{B})$ if there exist proper rigid transformations $(R_A,\mathbf{t}_A)$ and $(R_B,\mathbf{t}_B)$ such that:
	\begin{equation}\label{ajbra}
		\begin{tikzcd}[column sep=large]
		\mathcal{A}\arrow[r, "\mathcal{J}"]&\mathcal{B}\\
		R_A\mathcal{A}+\mathbf{t}_A\arrow[from=u, "{(R_A,\mathbf{t}_A)}"']\arrow[r, "\mathcal{J}'"']&R_B\mathcal{B}+\mathbf{t}_B\arrow[from=u, "{(R_B,\mathbf{t}_B)}"]
	\end{tikzcd}
	,\end{equation}
	i.e., $\mathcal{J}'=(R_B,\mathbf{t}_B)\mathcal{J}(R_A,\mathbf{t}_A)^{-1}$.
\end{dfn}
In general, two congruent CSMs have different domains and images, i.e., we have $\mathcal{J}' \notin \operatorname{CSM}(\mathcal{A}, \mathcal{B})$ in Eq.~(\ref{ajbra}), unless $(R_A, \mathbf{t}_A)$ maps $\mathcal{A}$ onto itself and $(R_B, \mathbf{t}_B)$ does the same.
In this section, we specifically focus on the case where the proper rigid transformations are symmetry operations of the crystal structures, which define a congruence relation on \(\operatorname{CSM}(\mathcal{A}, \mathcal{B})\).
We denote by $G_A$ the group formed by all proper rigid transformations $(R,\mathbf{t})$ such that
\begin{equation}\label{arta}
  \begin{tikzcd}
	  \mathcal{A}\arrow[r, "{(R,\mathbf{t})}"]&\mathcal{A}
  \end{tikzcd}
,\end{equation}
and define the point group
\begin{equation}\label{garrt}
	G_A' =\left\{R\,\middle|\,(R,\mathbf{t})\in G_A\right\}
.\end{equation}
It should be noted that $G_A'$ is \textit{not} a subgroup of $G_A$ unless the latter is symmorphic~\cite{dresselhaus2007group}, i.e., $G_A=G_A'\ltimes L_A$.
We will see that $G_A'$ and $G_B' $ induce an equivalence relation on $\operatorname{SLM}(\mathcal{A},\mathcal{B})$, which is inherited from Definition~\ref{def:cong}.
\begin{lem}\label{lem:inducedshuffle}
	If $\mathcal{J}\in\operatorname{CSM}(\mathcal{A},\mathcal{B})$ is a shuffle, then for any $\mathbf{t}_A,\mathbf{t}_B\in\mathbb{R}^3$, the composite mapping $(+\mathbf{t}_B)\circ\mathcal{J}\circ(-\mathbf{t}_A)$ is also a shuffle from $\mathcal{A}+\mathbf{t}_A$ to $\mathcal{B}+\mathbf{t}_B$, whose shuffle lattice is the same as $\mathcal{J}$.
\end{lem}
\begin{thm}\label{thm:sym}
	For any $(R_A,\mathbf{t}_A)\in G_A$, $(R_B,\mathbf{t}_B)\in G_B$, and $\mathcal{J}\in\operatorname{CSM}(\mathcal{A},\mathcal{B})$, the CSM $\mathcal{J}'$ defined by Eq.~(\ref{ajbra}) is also in $\operatorname{CSM}(\mathcal{A},\mathcal{B})$.
	If $\mathcal{J}$ has an SLM $(\tilde L_A,\tilde L_B,S)$, then $\mathcal{J}'$ has the SLM $(R_A\tilde L_A,R_B\tilde L_B,R_BSR_A^{-1})$.
\end{thm}
\begin{proof}[\indent Proof]
	When $(R_A,\mathbf{t}_A)\in G_A$ and $(R_B,\mathbf{t}_B)\in G_B$, the bottom line of Eq.~\ref{ajbra} becomes $\!\!\begin{tikzcd}\mathcal{A}\arrow[r, "\mathcal{J}'"]&\mathcal{B}\end{tikzcd}\!\!$.
	To show that $\mathcal{J}'$ has the SLM $(R_A\tilde L_A,R_B\tilde L_B,R_BSR_A^{-1})$, we only need to show that the deformed CSM
	\begin{align}
		&[(R_B,\mathbf{t}_B)\mathcal{J}(R_A,\mathbf{t}_A)^{-1}](R_BSR_A^{-1})^{-1}
		\\={}&(R_B,\mathbf{t}_B)\mathcal{J}R_A^{-1}\circ(-\mathbf{t}_A)\circ R_AS^{-1}R_B^{-1}
		\\={}&(R_B,\mathbf{t}_B)\mathcal{J}S^{-1}R_B^{-1}\circ(-R_BSR_A^{-1}\mathbf{t}_A)\label{rbtbj}
	\end{align}
	is a shuffle with shuffle lattice $R_B\tilde L_B$.
	By Definition~\ref{def:shuffle}, a vector $\mathbf{t}\in\mathbb{R}^3$ is in $\tilde L_B$ if and only if we have:
	\begin{equation}
		\begin{tikzcd}[column sep=large]
		  S\mathcal{A}\arrow[r, rightsquigarrow, "\mathcal{J}S^{-1}"]&\mathcal{B}\\
		  S\mathcal{A}\arrow[r, rightsquigarrow, "\mathcal{J}S^{-1}"']&\mathcal{B}
		  \arrow[from=1-1, to=2-1, "+\mathbf{t}"']
		  \arrow[from=1-2, to=2-2, "+\mathbf{t}"]
	  \end{tikzcd}
	.\end{equation}
	Applying $(R_B,\mathbf{t}_B)$ to each crystal structure, we obtain:
	\begin{equation}\label{zoom}
	\begin{tikzcd}[row sep=small, column sep=1.5em]
		R_BS\mathcal{A}+\mathbf{t}_B\arrow[from=rd, "{(R_B,\mathbf{t}_B)}", near start]&&&&R_B\mathcal{B}+\mathbf{t}_B\arrow[from=ld, "{(R_B,\mathbf{t}_B)}"', near start]\\
					      &S\mathcal{A}&&\mathcal{B}&\\&&&&\\
					      &S\mathcal{A}&&\mathcal{B}&\\
		R_BS\mathcal{A}+\mathbf{t}_B\arrow[from=ru, "{(R_B,\mathbf{t}_B)}"', near start]&&&&R_B\mathcal{B}+\mathbf{t}_B\arrow[from=lu, "{(R_B,\mathbf{t}_B)}", near start]
		\arrow[from=2-2, to=2-4, rightsquigarrow, "\mathcal{J}S^{-1}"]
		\arrow[from=4-2, to=4-4, rightsquigarrow, "\mathcal{J}S^{-1}"']
		\arrow[from=2-2, to=4-2, "+\mathbf{t}"']
		\arrow[from=2-4, to=4-4, "+\mathbf{t}"]
		\arrow[from=1-1, to=1-5, "{(R_B,\mathbf{t}_B)\mathcal{J}S^{-1}(R_B,\mathbf{t}_B)^{-1}}"]
		\arrow[from=5-1, to=5-5, "{(R_B,\mathbf{t}_B)\mathcal{J}S^{-1}(R_B,\mathbf{t}_B)^{-1}}"']
		\arrow[from=1-1, to=5-1, "+R_B\mathbf{t}"']
		\arrow[from=1-5, to=5-5, "+R_B\mathbf{t}"]
	\end{tikzcd}
	.\end{equation}
	Taking the outer loop of the above diagram and using $(R_B,\mathbf{t}_B)\in G_B$, we can see that $\mathbf{t}\in \tilde L_B$ if and only if:
	\begin{equation}
		\begin{tikzcd}[column sep=9.6em]
		  R_BS\mathcal{A}+\mathbf{t}_B\arrow[r, "{(R_B,\mathbf{t}_B)\mathcal{J}S^{-1}(R_B,\mathbf{t}_B)^{-1}}"]&\mathcal{B}\\
		  R_BS\mathcal{A}+\mathbf{t}_B\arrow[r, "{(R_B,\mathbf{t}_B)\mathcal{J}S^{-1}(R_B,\mathbf{t}_B)^{-1}}"']&\mathcal{B}
		  \arrow[from=1-1, to=2-1, "+R_B\mathbf{t}"']
		  \arrow[from=1-2, to=2-2, "+R_B\mathbf{t}"]
	  \end{tikzcd}
	,\end{equation}
	which implies that $(R_B,\mathbf{t}_B)\mathcal{J}S^{-1}(R_B,\mathbf{t}_B)^{-1}$ is a shuffle with shuffle lattice $R_B\tilde L_B$.
	Eq.~(\ref{rbtbj}) differs from this shuffle by only a translation, and is therefore also a shuffle with shuffle lattice $R_B\tilde L_B$ according to Lemma~\ref{lem:inducedshuffle}.
\end{proof}
\begin{dfn}[Congruent SLM]\label{def:slmcong}
	We say that an SLM is congruent to $(\tilde L_A,\tilde L_B,S)\in\operatorname{SLM}(\mathcal{A},\mathcal{B})$ if it can be written as $(R_A\tilde L_A,R_B\tilde L_B,R_BSR_A^{-1})$ for some $R_A\in G_A', R_B\in G_B' $.
\end{dfn}
Before describing how to exhaustively enumerate CSMs, we need to clarify what kind of CSM set is \textit{complete}.
We expect that they can, via NEB-like methods which allow rigid transformations, yield all SSPT mechanisms we are concerned with.
Hence, there is no need to enumerate two or more congruent CSMs, and we say that a subset of $X$ comprising CSMs or SLMs is complete as long as it contains an element of each congruence class in $X$.
\begin{thm}\label{thm:complete}
	If $X$ is a complete subset of $\operatorname{SLM}(\mathcal{A},\mathcal{B})$, then $\bigcup_{x\in X}Y_x$ is a complete subset of $\operatorname{CSM}(\mathcal{A},\mathcal{B})$ as long as $Y_x$ is a complete subset of those CSMs with SLM $x$.
\end{thm}
\begin{proof}[\indent Proof]
	As long as $\mathcal{J}\in\operatorname{CSM}(\mathcal{A},\mathcal{B})$ has an SLM, denoted by $(\tilde L_A,\tilde L_B,S)$, the completeness of $X$ ensures that there exist $R_A\in G_A' $ and $R_B\in G_B' $ such that
	\begin{equation}\label{ralar}
	  (R_A\tilde L_A,R_B\tilde L_B,R_BSR_A^{-1})\in X
	.\end{equation}
	Denote by $x$ the left-hand side of Eq.~(\ref{ralar}).
	Take $\mathbf{t}_A,\mathbf{t}_B$ such that $(R_A,\mathbf{t}_A)\in G_A$ and $(R_B,\mathbf{t}_B)\in G_B$.
	Since $\mathcal{J}$ is congruent to $(R_B,\mathbf{t}_B)\mathcal{J}(R_A,\mathbf{t}_A)^{-1}$ whose SLM is $x$ (Theorem~\ref{thm:sym}), the transitivity of congruence means that $\mathcal{J}$ is also congruent to some CSM in $Y_x$.
\end{proof}
Note that $\bigcup_{x\in X}Y_x$ is also a complete subset of
\begin{equation}\label{csmar}
	\bigcup_{\substack{R_1\in\operatorname{SO}(3),\boldsymbol{\tau}_1\in\mathbb{R}^3\\R_2\in\operatorname{SO}(3),\boldsymbol{\tau}_2\in\mathbb{R}^3}}\operatorname{CSM}(R_1\mathcal{A}+\boldsymbol{\tau}_1,R_2\mathcal{B}+\boldsymbol{\tau}_2)
,\end{equation}
which incorporates all possible initial and final structures of an SSPT.
Hence, the CSMs in $\bigcup_{x\in X}Y_x$ are sufficient as inputs to NEB-like methods.
We will demonstrate how to compute $X$ and $Y_x$ in Section~\ref{sec:algs}.
\section{Representation and pruning}\label{sec:rep-prun}
Before we describe the \textsc{crystmatch} method for exhaustively enumerating all noncongruent CSMs, it is necessary to represent CSMs as \textit{matrices} for computer processing.
We also need to introduce some physical constraints to make $\operatorname{CSM}(\mathcal{A},\mathcal{B})$ finite.
In this section, we will show how to represent SLM and shuffle separately, and discuss the general form of pruning criteria.
Afterwards, the finiteness of candidate CSMs will be proved.
Symbols frequently used in this section are summarized in Tables~\ref{tab:symbol-rep-prune}.
\subsection{Representing crystal-structure matches}\label{ssec:rep}
As elaborated in Section~\ref{sec:formal}, a CSM with finite multiplicity can always be decomposed into an SLM and a shuffle, both of which are cell-independent.
However, just as an element of a vector space acquires a coordinate representation only after a basis is chosen, to represent CSMs as matrices, we also need to select a pair of ``primitive cells'' for $\mathcal{A}$ and $\mathcal{B}$.
\subsubsection{Representing sublattice matches}\label{sssec:slm}
\begin{dfn}[Base Matrix]\label{def:bm}
	Let $C\in \mathbb{R}^{3\times 3}$ be a nonsingular matrix.
	We say that $ C(\mathbb{Z}^3)=\left\{C\mathbf{k}\,\middle|\,\mathbf{k}\in\mathbb{Z}^3\right\} $ is the lattice generated by $C$, and $C$ is the base matrix of $C(\mathbb{Z}^3)$.
\end{dfn}
\begin{lem}\label{lem:basis}
	Any lattice has a base matrix.
\end{lem}
\begin{lem}\label{lem:sublat}
	Let $C,\tilde C\in\mathbb{R}^{3\times 3}$ be nonsingular matrices.
	$\tilde C(\mathbb{Z}^3)$ is a sublattice of $C(\mathbb{Z}^3)$ if and only if there exists an $M\in\mathbb{Z}^{3\times 3}$ such that $\tilde C=CM$.
	In such case, the index of subgroup $\tilde C(\mathbb{Z}^3)$ in $C(\mathbb{Z}^3)$ is $\lvert \det M\rvert $.
\end{lem}
\begin{lem}\label{lem:quotient}
	Let $\mathcal{A}$ be a crystal structure and $\tilde L_A$ its sublattice.
	Denote by $k$ the index of subgroup $\tilde L_A$ in $L_A$, which is finite by Lemmas~\ref{lem:basis} and \ref{lem:sublat}.
	We have $\lvert \mathcal{A}/\tilde L_A\rvert=kZ_A$ as long as $Z_A=\lvert\mathcal{A} / L_A\rvert$ is finite.
\end{lem}
Take $C_A,C_B\in\mathbb{R}^{3\times 3}$ that generate $L_A$ and $L_B$, respectively.
For each SLM $(\tilde L_A,\tilde L_B,S)$, there exists an $M_A\in\mathbb{Z}^{3\times 3}$ such that $C_AM_A(\mathbb{Z}^3)=\tilde L_A$.
From Eq.~(\ref{slalb}), we can see that $SC_AM_A(\mathbb{Z}^3)=\tilde L_B$, so $SC_AM_A$ is a base matrix of $\tilde L_B$ which equals $C_BM_B$ for some $M_B\in\mathbb{Z}^{3\times 3}$.
This provides a way to represent an SLM $(\tilde L_A,\tilde L_B,S)$ by a pair of integer matrices $(M_A,M_B)$ satisfying
\begin{align}
	C_AM_A(\mathbb{Z}^3)&=\tilde L_A,\label{camaz}\\
	C_BM_B(\mathbb{Z}^3)&=\tilde L_B,\label{cbmbz}\\
	SC_AM_A&=C_BM_B,\label{scama}\\
	Z_A\lvert\det M_A\rvert&=Z_B\lvert\det M_B\rvert\label{zadma}
,\end{align}
where Eq.~(\ref{zadma}) is derived from Eq.~(\ref{alabl}) using Lemmas~\ref{lem:sublat} and \ref{lem:quotient}.
We can also see that any nonsingular $3\times 3$ integer matrix pair $(M_A, M_B)$ satisfying Eq.~(\ref{zadma}) determines an SLM according to Eqs.~(\ref{camaz}--\ref{scama}).
This representation of SLM is intuitive and extensively used~\cite{sheppard2012generalized,chen2016determination,li2022smallest,huang2024interface}, but is not suitable for SLM enumeration since many different $(M_A,M_B)$ can yield the same SLM~\cite{li2022smallest,wang2024crystal}.
To eliminate such representation redundancy, one may define a \textit{canonical base matrix} for each lattice.
The concepts of unimodular matrices and Hermite normal form will be useful, just as they play a crucial role in the Hart-Forcade theory~\cite{hart2008algorithm}.
\begin{dfn}[Unimodular Matrix]
	Let $Q\in\mathbb{Z}^{n\times n}$ be a nonsingular integer matrix. We say that $Q$ is unimodular if any of the following equivalent propositions holds:
	\begin{gather}
	Q(\mathbb{Z}^n)=\mathbb{Z}^n,\label{qznzn}\\
	Q^{-1}\in \mathbb{Z}^{n\times n},\\
	\det Q\in\{+1,-1\}
	.\end{gather}
	We denote the matrix group formed by all unimodular matrices as $\operatorname{GL}(n,\mathbb{Z})$.
\end{dfn}
\begin{table}[t]
	\caption{Frequently used symbols in Section~\ref{ssec:rep}.}\label{tab:symbol-rep-prune}
\begin{ruledtabular}
	\begin{tabular}{cl}
		Symbol & Meaning \\ \hline
		$\mathbb{F}^{n\times m}$ & All $n\times m$ matrices over the field $\mathbb{F}$\\
		$C(\mathbb{Z}^3)$ & Lattice generated by nonsingular $C\in\mathbb{R}^{3\times 3}$ \\
		$C_A$ & Base matrix of $L_A$ (primitive-cell vectors) \\
		$\operatorname{GL}(n,\mathbb{Z})$ & All $n\times n$ unimodular matrices\\
		$\operatorname{hnf}(M)$ & HNF of the integer matrix $M$ \\
		$\operatorname{HNF}(k)$ & All $3\times 3$ HNFs with determinant $k$ \\
		$(H_A,H_B,Q)$ & IMT representation of an SLM \\
		$\operatorname{IMT}(\mu)$ & All IMTs with multiplicity $\mu$ \\
		$p$ & Permutation on the set $\{1,\cdots,\tilde Z\}$ \\
		$(p,\mathbf{t}_1,\cdots,\mathbf{t}_{\tilde Z})$ & PCT representation of a shuffle \\
		$w(S)$ & Estimated strain energy density of $S$ \\
		$\sigma_i(S), s_i$ & $i$-th largest singular value of $S$ \\
		$\tilde{\mathcal{J}}$ & Standard shuffle of $\mathcal{J}$ (generally $\tilde{\mathcal{J}}\neq\mathcal{J}S^{-1}$) \\
		$\tilde{\mathcal{A}},\tilde{\mathcal{B}}$ & Half-distorted structures with sublattice $\tilde L$\\
		$\tilde L$ & Shuffle lattice of $\tilde{\mathcal{J}}$ \\
		$\theta_i$ & Normalized weight of the $i$-th atom \\
		$\ell$ & Type of norm used to define shuffle distance \\
		$\hat{d}(\mathcal{J})$ & Shuffle distance of $\mathcal{J}$ \\
		$d(\mathcal{J})$ & Minimum $\hat{d}$ among all CSMs congruent to $\mathcal{J}$
\end{tabular}
\end{ruledtabular}
\end{table}
\begin{lem}\label{lem:redundancy}
	Two nonsingular matrices $C,C'\in\mathbb{R}^{3\times 3}$ generate the same lattice if and only if $C^{-1}C'\in\operatorname{GL}(3,\mathbb{Z})$.
\end{lem}
\begin{dfn}[HNF]\label{def:hnf}
	Let $H\in \mathbb{Z}^{n\times n}$ be a nonsingular matrix.
	We say that $H$ is in Hermite normal form (HNF) if (1) $H$ is lower triangular and (2) each off-diagonal element of $H$ is nonnegative and \textit{strictly} less than the diagonal element in its row.
	In other words, the elements of $H$ satisfy
	\begin{equation}
	\begin{cases}
		H_{ii} > H_{ij} \ge 0,&\text{if } j<i\\
		H_{ij}=0,&\text{if } j>i
	\end{cases}
	\end{equation}
\end{dfn}
\begin{lem}\label{lem:hnf}
	Let $M\in \mathbb{Z}^{n\times n}$ be a nonsingular matrix.
	There exists a unique $H\in \mathbb{Z}^{n\times n}$ in HNF and a unique $Q\in\operatorname{GL}(n,\mathbb{Z})$ such that
	\begin{equation}\label{qglmz}
		\exists Q\in\operatorname{GL}(n,\mathbb{Z}),\quad M=HQ
	.\end{equation}
	We say that $H$ is the HNF of $M$.
\end{lem}
Let $\operatorname{hnf}(M)$ denote the HNF of $M$, and $\operatorname{HNF}(k)$ denote all $3\times 3$ HNFs with determinant $k$.
Lemma~\ref{lem:redundancy} implies that $(M_A,M_B)$ and $(M_A',M_B')$ represent the same SLM if and only if $(M_A',M_B')=(M_AQ',M_BQ')$ for some $Q'\in\operatorname{GL}(3,\mathbb{Z})$.
To make the representation of SLM unique, one may use the mapping $(M_A,M_B)\mapsto (H_A,H_B,Q)$, where
\begin{align}
	H_A&=\operatorname{hnf}(M_A),\\
	H_B&=\operatorname{hnf}(M_B),\\
	Q&=(H_B^{-1}M_B)(H_A^{-1}M_A)^{-1}
.\end{align}
It maps the entire $\{(M_AQ',M_BQ')\,|\,Q'\in\operatorname{GL}(3,\mathbb{Z})\}$ to a \textit{single} integer-matrix triplet (IMT)~\cite{wang2024crystal}.
\begin{dfn}[IMT]\label{def:imt}
	Let $\mathcal{A}$, $\mathcal{B}$ be crystal structures and $\mu$ a positive integer.
	We say that
	\begin{multline}
		\operatorname{IMT}(\mu)=\operatorname{HNF}\!\left(\frac{\mu\operatorname{lcm}(Z_A,Z_B)}{Z_A}\right)\\
		\times\operatorname{HNF}\!\left(\frac{\mu\operatorname{lcm}(Z_A,Z_B)}{Z_B}\right)\times\operatorname{GL}(3,\mathbb{Z})
	\end{multline}
	is the set of all IMTs with multiplicity $\mu$.
\end{dfn}
\begin{thm}\label{thm:imt}
	The mapping $(H_A,H_B,Q)\mapsto(\tilde L_A,\tilde L_B,S)$ with
	\begin{align}
		\tilde L_A&=C_AH_A(\mathbb{Z}^3),\label{lacah}\\
		\tilde L_B&=C_BH_B(\mathbb{Z}^3),\label{lacbh}\\
		S&=(C_BH_BQ)(C_AH_A)^{-1}\label{scbhb}
	\end{align}
	is a bijection from $\operatorname{IMT}(\mu)$ to $\operatorname{SLM}(\mathcal{A},\mathcal{B};\mu)$.
\end{thm}
\begin{proof}[\indent Proof]
	$(\tilde L_A,\tilde L_B,S)$ is an SLM since $S\tilde L_A=\tilde L_B$ and
	\begin{align}
		\lvert \mathcal{A}/\tilde L_A\rvert
		&= Z_A \det H_A \label{zadha}\\
		&= \mu\operatorname{lcm}(Z_A,Z_B) \label{mlcmz}\\
		&= Z_B \det H_B \label{zbdhb}\\
		&=\lvert \mathcal{B}/\tilde L_B\rvert \label{blb}
	,\end{align}
	where Lemmas~\ref{lem:sublat} and \ref{lem:quotient} are used to derive Eqs.~(\ref{zadha}) and (\ref{blb}).
	Definition~\ref{def:imt} ensures Eqs.~(\ref{mlcmz}) and (\ref{zbdhb}), so that the multiplicity of $(\tilde L_A,\tilde L_B,S)$ is $\mu$.
	Conversely, for each $(\tilde L_A,\tilde L_B,S)\in \operatorname{SLM}(\mathcal{A},\mathcal{B};\mu)$, $\tilde L_A$ has a unique base matrix $C_AH_A$, where $H_A$ is in HNF (Lemmas~\ref{lem:redundancy} and \ref{lem:hnf}) with determinant $\tilde Z / Z_A = \mu$ (Lemmas~\ref{lem:sublat} and \ref{lem:quotient}), and so does $\tilde L_B$.
	Since $H_A$ and $H_B$ are determined, $S$ further determines $Q$ by Eq.~(\ref{scbhb}).
	Therefore, every SLM is associated with a unique IMT.
\end{proof}
Theorem~\ref{thm:imt} allows us to focus on the enumeration of $\operatorname{IMT}(\mu)$ from now on, which is far more intuitive than the enumeration of SLMs.
\subsubsection{Representing shuffles}
\begin{dfn}[Motif]\label{def:motif}
	Let $\mathcal{A}$ be a crystal structure and $\tilde L_A$ a sublattice of $\mathcal{A}$.
	We say that $\{\mathbf{a}_1,\mathbf{a}_2\cdots,\mathbf{a}_n\}\subset\mathcal{A}$ is an $\tilde L_A$-motif of $\mathcal{A}$ if $\mathbf{a}_1,\cdots,\mathbf{a}_n$ are $\tilde L_A$-inequivalent and $n=\lvert \mathcal{A}/\tilde L_A\rvert $.
\end{dfn}
\begin{thm}\label{thm:pct}
	Let $(\tilde L_A,\tilde L_B,S)$ be an SLM with period $\tilde Z$, $\{\mathbf{a}_1,\cdots,\mathbf{a}_{\tilde Z}\}$ an $\tilde L_A$-motif of $\mathcal{A}$, and $\{\mathbf{b}_1,\cdots,\mathbf{b}_{\tilde Z}\}$ an $\tilde L_B$-motif of $\mathcal{B}$.
	If $\mathcal{J}\in\operatorname{CSM}(\mathcal{A},\mathcal{B})$ has the SLM $(\tilde L_A,\tilde L_B,S)$, there exist a unique permutation $p\in\operatorname{Sym}(\tilde Z)$ and unique $\mathbf{t}_1,\cdots,\mathbf{t}_{\tilde Z}\in\tilde L_B$ such that
	\begin{equation}\label{jaibp}
		\forall i\in\{1,\cdots,\tilde Z\},\quad\mathcal{J}(\mathbf{a}_i)=\mathbf{b}_{p(i)}+\mathbf{t}_{i}
	.\end{equation}
	Furthermore, if both $\mathcal{J}_1$ and $\mathcal{J}_2$ have $(\tilde L_A,\tilde L_B,S)$, then $\mathcal{J}_1=\mathcal{J}_2$ if and only if they have the same $(p,\mathbf{t}_1,\cdots,\mathbf{t}_{\tilde Z})$.
\end{thm}
\begin{proof}[\indent Proof]
	In the proof of Theorem~\ref{thm:slm}, we have already shown that $\mathbf{a}+\tilde L_A\mapsto\mathcal{J}(\mathbf{a})+\tilde L_B$ is a bijection from $\mathcal{A}/\tilde L_A$ to $\mathcal{B}/\tilde L_B$.
	Therefore, the images of different $\mathbf{a}_1,\cdots,\mathbf{a}_{\tilde Z}$ under $\mathcal{J}$ are in different $\mathbf{b}_1+\tilde L_B,\cdots,\mathbf{b}_{\tilde Z}+\tilde L_B$.
	This determines $p$ which is a permutation on $\{1,\cdots,\tilde Z\}$, and $\mathbf{t}_1,\cdots,\mathbf{t}_{\tilde Z}$ is also uniquely determined by Eq.~(\ref{jaibp}).
	Furthermore, since each element in $\mathcal{A}$ can be uniquely decomposed as $\mathbf{a}=\mathbf{a}_i+\mathbf{t}$ for some $i\in\{1,\cdots,\tilde Z\}$ and $\mathbf{t}\in\tilde L_A$, we have
	\begin{equation}\label{aitaj}
		\forall \mathbf{a}_i+\mathbf{t}\in\mathcal{A},\quad \mathcal{J}(\mathbf{a}_i+\mathbf{t})=\mathbf{b}_{p(i)}+\mathbf{t}_i+S\mathbf{t}
	,\end{equation}
	which is derived from Eq.~(\ref{jaibp}) and the following commutative diagram:
	\begin{equation}
		\begin{tikzcd}[column sep=large]
		  \mathcal{A}\arrow[r, "S"']&S\mathcal{A}\arrow[r, "\mathcal{J}S^{-1}"',rightsquigarrow]&\mathcal{B}\\
		  \mathcal{A}\arrow[r, "S"]&S\mathcal{A}\arrow[r, "\mathcal{J}S^{-1}",rightsquigarrow]&\mathcal{B}
		  \arrow[from=1-1, to=2-1, "+\mathbf{t}"']
		  \arrow[from=1-2, to=2-2, "+S\mathbf{t}"']
		  \arrow[from=1-3, to=2-3, "+S\mathbf{t}"]
		  \arrow[from=1-1, to=1-3, "\mathcal{J}", bend left=25]
		  \arrow[from=2-1, to=2-3, "\mathcal{J}"', bend right=25]
	  \end{tikzcd}
	.\end{equation}
	From Eq.~(\ref{aitaj}) we can see that $(p,\mathbf{t}_1,\cdots,\mathbf{t}_{\tilde Z})$ completely determines a CSM with the SLM $(\tilde L_A,\tilde L_B,S)$.
\end{proof}
Given the SLM, Theorem~\ref{thm:pct} uniquely represents each CSM as a permutation $p$ and $\tilde Z$ elements of $\tilde L_B$.
The underlying reason is that each $(p,\mathbf{t}_1,\cdots,\mathbf{t}_{\tilde Z})$ actually represents a shuffle, which we refer to as the ``permutation with class-wise translations'' (PCT) representation, as illustrated in Fig.~\ref{fig:pct}.
We denote by $\mathcal{J}_{p,\mathbf{t}_1,\cdots,\mathbf{t}_{\tilde Z}}$ the CSM determined by Eq.~(\ref{aitaj}).
\begin{figure}[t!]
  \centering
  \includegraphics[width=0.95\linewidth]{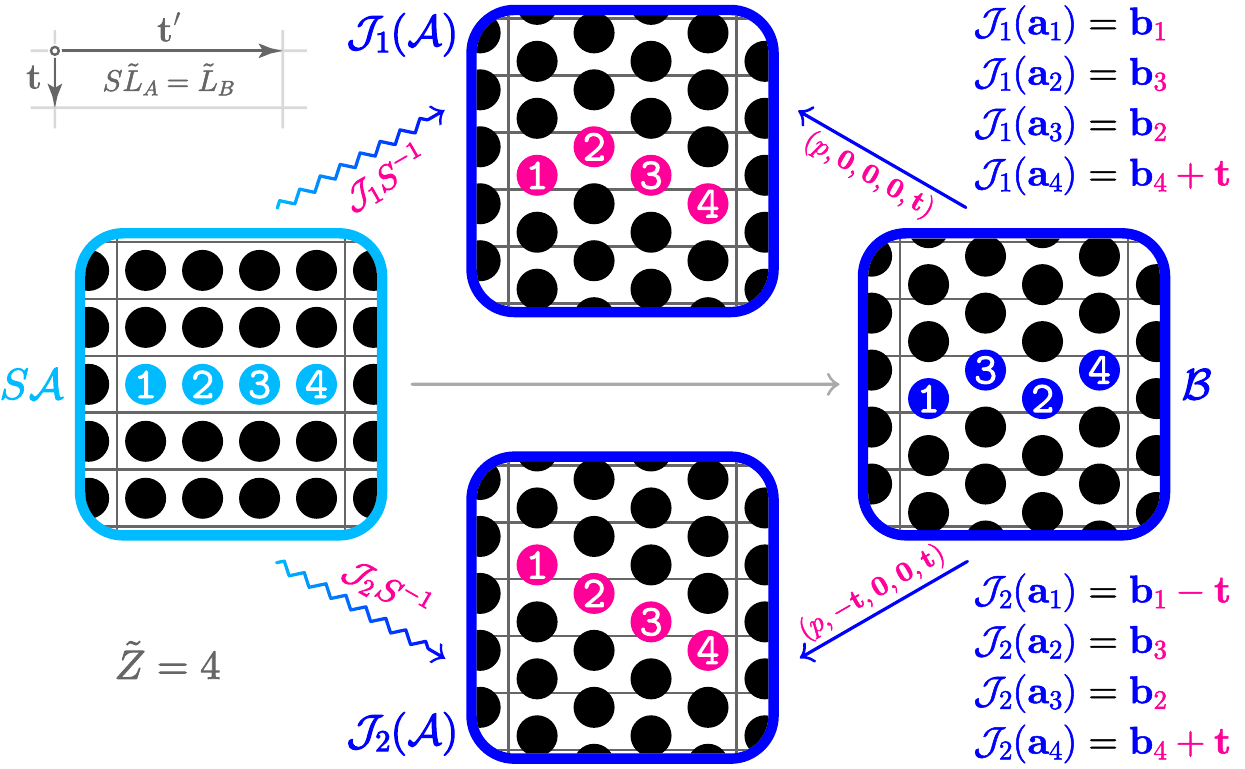}
  \caption{Two PCTs representing $\mathcal{J}_1,\mathcal{J}_2\colon S\mathcal{A}\rightsquigarrow\mathcal{B}$ with the same shuffle lattice (gray grids), respectively.
	  Both $\{S\mathbf{a}_i\}$, the motif of $S\mathcal{A}$ (cyan) and $\{\mathbf{b}_i\}$, the motif of $\mathcal{B}$ (blue) are labeled by \texttt{1-4}.
	  Every pink atom $\mathcal{J}(\mathbf{a}_i)$ is both the image of $S\mathbf{a}_i$ under $\mathcal{J}S^{-1}$ and an element of $\mathcal{B}$, which can be written as $\mathbf{b}_{p(i)}+\mathbf{t}_i$.
	  Therefore, the PCT representation of a shuffle can be viewed as the unique map $\mathbf{b}_i\mapsto\mathbf{b}_{p(i)}+\mathbf{t}_i$ that makes the triangle in the figure commute, where the horizontal gray arrow represents $S\mathbf{a}_i\mapsto \mathbf{b}_i$ (the canonical bijection between motifs).
	  The pink text shows the specific PCTs representing $\mathcal{J}_1$ and $\mathcal{J}_2$, where $p=(\begin{smallmatrix}
		  1&2&3&4\\1&3&2&4
	  \end{smallmatrix})$.
  }
  \label{fig:pct}
\end{figure}
\begin{dfn}[PCT]\label{def:pct}
	Given a mapping
	\begin{equation}\label{lalbs}
	  (\tilde L_A, \tilde L_B, S)\mapsto(\mathbf{a}_1,\cdots,\mathbf{a}_{\tilde Z},\mathbf{b}_1,\cdots,\mathbf{b}_{\tilde Z})
	,\end{equation}
	that determines an $\tilde L_A$-motif of $\mathcal{A}$ and an $\tilde L_B$-motif of $\mathcal{B}$ for each SLM in $\operatorname{SLM}(\mathcal{A}, \mathcal{B})$.
	We say that
	\begin{multline}
		\operatorname{PCT}(\tilde L_A,\tilde L_B,S)=\Big\{(p,\mathbf{t}_1,\cdots,\mathbf{t}_{\tilde Z})\in\operatorname{Sym}(\tilde Z)\times \tilde L_B^{\tilde Z}\,\Big|
		\\\forall i,\,\chi_A(\mathbf{a}_i)=\chi_B(\mathbf{b}_{p(i)})\Big\}
	\end{multline}
	is the set of all PCTs with SLM $(\tilde L_A,\tilde L_B,S)$, which depends only on the SLM and the mapping.
\end{dfn}
Note that $p$ can be stored as a permutation matrix, and $\tilde L_B$ is isomorphic to $\mathbb{Z}^3$.
In this sense, all the ingredients of a CSM can be represented as integer matrices.
Given $C_A$, $C_B$, and a deterministic way to choose the $\tilde L_A$- and $\tilde L_B$-motifs for each SLM, one can reproduce any CSM from its associated IMT and PCT.
\subsection{Pruning criteria}\label{ssec:prune}
As a one-to-one correspondence between infinite sets, the CSM has infinitely many possibilities.
Fortunately, we are only concerned with the CSMs of those MEPs with low energy barriers.
Hence, many CSMs can be excluded, leaving only a finite number of candidates.
\subsubsection{Strain energy density}\label{sssec:strain}
From an energetic perspective, the constraint on strain is necessary, since excessive deformation would lead to a large strain energy.
We are only concerned with those CSMs whose deformation gradients satisfy
\begin{equation}\label{wsdth}
	w(S) \le  w_\text{max}
,\end{equation}
where $w(S)$ denotes the strain energy density associated with the deformation gradient $S$, and $w_\text{max}$ is a truncation threshold.
When searching for the CSM of the lowest-barrier MEP, one may let $w_\text{max}$ be the minimum $\Delta^\ddagger$-density among known SSPTs.
\begin{lem}[SVD]\label{lem:svd}
	If $S\in \mathbb{R}^{3\times 3}$ has a positive determinant, there exist $U,V\in\operatorname{SO}(3)$ and unique $s_1\ge s_2\ge s_3>0$ such that $S=U\Sigma V^\text{T}$, where $\Sigma=\operatorname{diag}(s_1,s_2,s_3)$.
\end{lem}
Lemma~\ref{lem:svd} is known as the singular value decomposition (SVD) and $s_i$ is referred to as the $i$-th singular value of $S$~\cite{axler2015linear}, usually denoted by $\sigma_i(S)$.
A direct corollary is the polar decomposition (PD), which states that a deformation gradient can always be \textit{uniquely} factorized into a positive-definite transformation and a rotation, as
\begin{equation}
	S=(UV^{\text{T}})\left(V\Sigma V^{\text{T}}\right)=\left(S\sqrt{S^\text{T}S}^{\,-1}\right)\sqrt{S^\text{T}S}\label{suvtv}
.\end{equation}
Note that the rotation $UV^{\text{T}}$ does not contribute to the strain energy, and the positive-definite $V\Sigma V^{\text{T}}$ simply scales the space by $s_i$ along $\mathbf{v}_i$ for $i=1,2,3$, where $\mathbf{v}_i$ denotes the $i$-th column of $V$.
We estimate the strain energy according to the diagram:
\begin{equation}\label{ajbvs}
  \begin{tikzcd}[column sep=1.8em]
	  \mathcal{A}\arrow[rr, "\mathcal{J}"]&&\mathcal{B}\\
	  V\Sigma^{\frac{1}{2}} V^\text{T}\mathcal{A}\arrow[r, rightsquigarrow, "\tilde{\mathcal{J}}"']&V\Sigma^{-\frac{1}{2}}U^\text{T}\mathcal{B}\arrow[r, "UV^\text{T}"']&U\Sigma^{-\frac{1}{2}} U^\text{T}\mathcal{B} 
	  \arrow[from=1-1, to=2-1, "\sqrt{V\Sigma V^\text{T}}"']
	  \arrow[from=2-3, to=1-3, "\sqrt{U\Sigma U^\text{T}}"']
  \end{tikzcd}
\end{equation}
where $\tilde{\mathcal{J}}$ is defined as the unique shuffle that makes the diagram commute, which we refer to as the \textit{standard shuffle} of $\mathcal{J}$.
Note that the two vertical arrows are pure deformations, while all horizontal arrows in the second row involve no strain.
Based on this, $w$ is estimated as the sum of the strain energies associated with the deformation of $\mathcal{A}$ by $\sqrt{V\Sigma V^\text{T}}$ and the deformation of $\mathcal{B}$ by $\sqrt{U\Sigma U^\text{T}}^{\,-1}$, namely,
\begin{widetext}
\begin{align}
	w(S)&=\frac{1}{2}\sum_{i,j,k,l=1}^3\left\{Y_{ijkl}^A \left[V(\Sigma^{\frac{1}{2}}-I)V^\text{T}\right]_{ij} \left[V(\Sigma^{\frac{1}{2}}-I)V^\text{T}\right]_{kl}
	+Y_{ijkl}^B \left[U(\Sigma^{-\frac{1}{2}}-I)U^\text{T}\right]_{ij} \left[U(\Sigma^{-\frac{1}{2}}-I)U^\text{T}\right]_{kl}\right\}\label{w12ij}
	\\
	 &=\frac{1}{2}\sum_{a,b=1}^3\sum_{i,j,k,l=1}^3\left[\left(\sqrt{s_a} -1\right)\left(\sqrt{s_b} -1\right)Y_{ijkl}^AV_{ia}V_{ja}V_{kb}V_{lb}
		 +\left(\frac{1}{\sqrt{s_a}} -1\right)\left(\frac{1}{\sqrt{s_b}} -1\right)Y_{ijkl}^BU_{ia}U_{ja}U_{kb}U_{lb}\right]\label{12ab1}
,\end{align}
\end{widetext}
where $Y_{ijkl}^{A}$ and $Y_{ijkl}^{B}$ denote the elastic (stiffness) tensors of $\mathcal{A}$ and $\mathcal{B}$, respectively.
%
%For isotropic materials, we have~\cite{marsden1994mathematical}
%\begin{gather}
%	Y_{ijkl}=\left(K-\frac{2}{3}G\right)\delta_{ij}\delta_{kl}+G(\delta_{ik}\delta_{jl}+\delta_{il}\delta_{jk}),\\
%	w=G\sum_{i=1}^3\varepsilon_i^2+\left(\frac{K}{2}-\frac{G}{3}\right)\left(\sum_{i=1}^3\varepsilon_i\right)^2\label{wgi13}
%,\end{gather}
%where $K,G>0$ are the bulk and shear moduli.
%

%
Substituting Eq.~(\ref{12ab1}) into Eq.~(\ref{wsdth}) yields the inequality that $S$ must satisfy, which varies depending on the material's mechanical properties.
It is worth noting that $w$ is a \textit{congruence-class function}, i.e.,
\begin{equation}
	\forall R_A\in G_A',\quad\forall R_B\in G_B',\quad w(R_BSR_A^{-1})=w(S)
.\end{equation}
This can be seen from Lemma~\ref{lem:svd}, where $R_A$ and $R_B$ only contributes to $V$ and $U$, respectively.
As symmetry operations of $\mathcal{A}$ and $\mathcal{B}$, the elastic tensors $Y_{ijkl}^A$ and $Y_{ijkl}^B$ are invariant under $R_A$ and $R_B$, respectively.
Hence, SLMs are pruned by congruence classes under Eq.~(\ref{wsdth}).
\subsubsection{Shuffle distance}\label{sssec:rmsd}
For each SLM, there are still infinitely many possible CSMs---even a single atom $\mathbf{a}\in\mathcal{A}$ has infinitely many choices of its counterpart $\mathcal{J}(\mathbf{a})$.
An intuitive constraint is that $\mathcal{J}(\mathbf{a})$ should not be too far from $\mathbf{a}$.
This is not an energetic consideration, but rather an empirical one, as discussed in Ref.~\citenum{therrien2020minimization}, where the average distance traveled by the atoms is shown to be a key indicator.
However, the average distance between $\mathcal{J}(\mathbf{a})$ and $\mathbf{a}$ goes to infinity unless $S=I$, i.e., $\mathcal{J}$ is a shuffle.
Therefore, the appropriate way to define the ``average distance'' is to consider a shuffle associated with $\mathcal{J}$.
Here, we use the standard shuffle $\tilde{\mathcal{J}}$ defined in Eq.~(\ref{ajbvs}) instead of $\mathcal{J}S^{-1}$ to ensure the metric symmetry between the initial and final structures.
As shorthands, denote by $\tilde L$ the shuffle lattice of $\tilde{\mathcal{J}}$ and let $\tilde{\mathcal{A}}=V\Sigma^{\frac{1}{2}}V^\text{T}\mathcal{A}$ and $\tilde{\mathcal{B}}=V\Sigma^{-\frac{1}{2}}U^\text{T}\mathcal{B}$, which are all uniquely determined by the SLM.
\begin{dfn}[Shuffle Distance]\label{def:rmsd}
	Let $\mathcal{J}\in\operatorname{CSM}(\mathcal{A},\mathcal{B})$ be a CSM with period $\tilde Z$.
	We say that the shuffle distance of $\mathcal{J}$ is
	\begin{equation}\label{dj1zi}
		\hat{d}(\mathcal{J})=\left(\sum_{i=1}^{\tilde Z}\theta_i\left|\tilde{\mathcal{J}}(\tilde{\mathbf{a}}_i)-\tilde{\mathbf{a}}_i\right|^\ell\right)^{1 / \ell}
	,\end{equation}
	where $\ell\ge 1$ specifies the norm, $\{\tilde{\mathbf{a}}_1,\cdots,\tilde{\mathbf{a}}_{\tilde Z}\}$ is an $\tilde L$-motif of $\tilde{\mathcal{A}}$, and $\theta_i>0$ is a normalized weight that only depends on the atomic species of $\tilde{\mathbf{a}}_i$ (e.g., a constant, or atomic mass).
	One may adopt $\ell=2$ for the root-mean-square distance (RMSD), or $\ell=1$ for the arithmetic mean distance.
\end{dfn}
Eq.~\ref{ajbtt} ensures the value of Eq.~(\ref{dj1zi}) to be independent of the choice of the motif so that $\hat{d}$ is well-defined.
However, congruent CSMs may not necessarily have the same $\hat{d}$, which is undesirable since they are equivalent as inputs to NEB-like methods.
To address this issue, we define
\begin{equation}\label{djmin}
	d(\mathcal{J})=\min_{\boldsymbol{\tau}\in\mathbb{R}^3}\left(\sum_{i=1}^{\tilde Z}\theta_i\lvert \tilde{\mathcal{J}}(\tilde{\mathbf{a}}_i)+\boldsymbol{\tau}-\tilde{\mathbf{a}}_i\rvert^\ell \right)^{1 / \ell}
,\end{equation}
which equals the minimum $\hat{d}$ among all CSMs of the form $(+\boldsymbol{\tau}_B)\circ\mathcal{J}\circ(-\boldsymbol{\tau}_A)\in\operatorname{CSM}(\mathcal{A}+\boldsymbol{\tau}_A,\mathcal{B}+\boldsymbol{\tau}_B)$.
Similar to the proof of Theorem~\ref{thm:sym}, one can show that any $\mathcal{J}'$ congruent to $\mathcal{J}$ has a standard shuffle that differs from $(R,\mathbf{t})\tilde{\mathcal{J}}(R,\mathbf{t})^{-1}$ by just a translation, implying that
\begin{equation}\label{djdrt}
	d(\mathcal{J}')
	=d[(R,\mathbf{t})\tilde{\mathcal{J}}(R,\mathbf{t})^{-1}]
	=d(\tilde{\mathcal{J}})=d(\mathcal{J})
,\end{equation}
where we have used the fact that $(R,\mathbf{t})$ is an isometry.
Eq.~(\ref{djdrt}) means that congruent CSMs have the same $d$, so that as long as $d(\mathcal{J})$ exceeds the threshold $d_\text{max}$, we can exclude all CSMs congruent to $\mathcal{J}$.
\subsubsection{Finiteness of candidate crystal-structure matches}
As shown in Fig.~\ref{fig:tree}, the enumeration includes a CSM $\mathcal{J}$ with multiplicity $\mu$ and deformation gradient $S$ if
\begin{align}
	\mu&\le \mu_\text{max},\label{mmmax}\\
	w(S)&\le w_\text{max},\label{wswma}\\
	d(\mathcal{J})&\le d_\text{max}\label{djdma}
.\end{align}
They have all been detailed earlier, except for Eq.~(\ref{mmmax}), which can be regarded as an empirical constraint inferred from Table~\ref{tab:mech}---despite MD and MetaD simulations using far more than $10^3$ atoms, the intricacy of the CSMs they produce remain $\tilde{Z} \sim 10^1$.
We emphasize that $w_\text{max}$ can depend on $\mu$, and $d_\text{max}$ can depend on the SLM of $\mathcal{J}$.
For example, one might want to slacken the constraints on $w_\text{max}$ and $d_\text{max}$ when $\mu$ is small.
One may also set $d_\text{max}$ for each SLM to $110\%$ of the lowest $d$ value among all CSMs with that SLM.
In a word, different subtrees in Fig.~\ref{fig:tree} can be pruned with different thresholds.
We denote by $\operatorname{CSM}^\star (\mathcal{A},\mathcal{B})$ the set of all CSMs in $\operatorname{CSM}(\mathcal{A}, \mathcal{B})$ that satisfy Eqs.~(\ref{mmmax}--\ref{djdma}), and their SLMs by $\operatorname{SLM}^\star (\mathcal{A},\mathcal{B})$.

\begin{thm}\label{thm:finite-slm}
	$\operatorname{SLM}^\star (\mathcal{A},\mathcal{B};\mu)$ is finite regardless of the values of $\mu$ and $w_\text{max}$.
\end{thm}
\begin{proof}[\indent Proof]
	There exists a $\sigma_\text{max}$ such that $\sigma_1(S) < \sigma_\text{max}$ implies Eq.~(\ref{wswma}); otherwise, infinite strain would leads to infinite strain energy.
	We only need to prove that only a finite number of $(H_A,H_B,Q)\in\operatorname{IMT}(\mu)$ satisfy $\sigma_1(S)<\sigma_\text{max}$.
	Since $\operatorname{HNF}(k)$ is finite, the number of possible values of $(H_A,H_B)$ is already finite.
	It is sufficient to show that given $(H_A,H_B)$, the number of possible values of $Q$ is finite.
	This is due to the fact that the absolute value of any matrix element of $Q$ does not exceed
	\begin{align}
		\sigma_1(Q)
		&=\sigma_1(H_B^{-1}C_B^{-1}SC_AH_A)\label{s1qs1}
	      \\&\le\sigma_1(H_B^{-1})\sigma_1(C_B^{-1})\sigma_1(C_A)\sigma_1(H_A)\,\sigma_\text{max}
	,\end{align}
	where Eq.~(\ref{s1qs1}) is derived from Eq.~(\ref{scbhb})
\end{proof}
\begin{thm}\label{thm:finite-shuffle}
	For any given SLM and $d_\text{max}$, the number of noncongruent CSMs that have this SLM and satisfy Eq.~(\ref{djdma}) is finite.
\end{thm}
\begin{proof}[\indent Proof]
	We assume by contradiction that there are infinitely many noncongruent CSMs in $\operatorname{CSM}^\star (\mathcal{A},\mathcal{B})$ that have $(\tilde L_A,\tilde L_B,S)$.
	Among them, each $\mathcal{J}_{p,\mathbf{t}_1,\mathbf{t}_2,\cdots,\mathbf{t}_{\tilde Z}}$ is congruent to a $\mathcal{J}_{p,\mathbf{0},\mathbf{t}_2',\cdots,\mathbf{t}_{\tilde Z}'}$, where $\mathbf{t}_i'=\mathbf{t}_i-\mathbf{t}_1$.
	Since $p$ can take only a finite number of values, there are infinitely many PCTs having the same $p$ and different $(\mathbf{0},\mathbf{t}_2',\cdots,\mathbf{t}_{\tilde Z}')$ that satisfy
	\begin{equation}\label{djp0t}
		d(\mathcal{J}_{p,\mathbf{0},\mathbf{t}_2',\cdots,\mathbf{t}_{\tilde Z}'})\le d_\text{max}
	.\end{equation}
	These $(\mathbf{0},\mathbf{t}_2',\cdots,\mathbf{t}_{\tilde Z}')$ form an infinite subset of $\tilde L_B^{\tilde Z}$, so there exist $j$ such that $\lvert \mathbf{t}_j'\rvert$ is unbounded; otherwise, the subset would be finite.
	Let
	\begin{equation}
	  y=\left|\tilde{\mathbf{b}}_{p(j)}+\tilde{\mathbf{t}}_j'-\tilde{\mathbf{b}}_{p(1)}-\tilde{\mathbf{a}}_j+\tilde{\mathbf{a}}_1\right|
	,\end{equation}
	which is also unbounded. We have 
	\begin{widetext}
	\begin{align}
		[d(\mathcal{J}_{p,\mathbf{0},\mathbf{t}_2',\cdots,\mathbf{t}_{\tilde Z}'})]^\ell
		&=\min_{\boldsymbol{\tau}\in\mathbb{R}^3}\sum_{i=1}^{\tilde Z}\theta_i\left| \tilde{\mathbf{b}}_{p(i)}+\tilde{\mathbf{t}}_i'+\boldsymbol{\tau}-\tilde{\mathbf{a}}_i\right|^\ell
	      \\&\ge\min_{\boldsymbol{\tau}\in\mathbb{R}^3}\left(\theta_1\left| \tilde{\mathbf{b}}_{p(1)}+\boldsymbol{\tau}-\tilde{\mathbf{a}}_1\right|^\ell + \theta_j\left| \tilde{\mathbf{b}}_{p(j)}+\tilde{\mathbf{t}}_j'+\boldsymbol{\tau}-\tilde{\mathbf{a}}_j\right|^\ell\right)
	      \\&=\min_{\boldsymbol{\tau}\in\mathbb{R}^3}\left(\theta_1\left| \tilde{\mathbf{b}}_{p(1)}+\boldsymbol{\tau}-\tilde{\mathbf{a}}_1\right|^\ell + \theta_j\left| \left(\tilde{\mathbf{b}}_{p(1)}+\boldsymbol{\tau}-\tilde{\mathbf{a}}_1\right)+\left(\tilde{\mathbf{b}}_{p(j)}+\tilde{\mathbf{t}}_j'-\tilde{\mathbf{b}}_{p(1)}-\tilde{\mathbf{a}}_j+\tilde{\mathbf{a}}_1\right)\right|^\ell\right)
	      \\&\ge\min_{\boldsymbol{\tau}\in\mathbb{R}^3}\left(\theta_1\left| \tilde{\mathbf{b}}_{p(1)}+\boldsymbol{\tau}-\tilde{\mathbf{a}}_1\right|^\ell + \theta_j\bigg| \left|\tilde{\mathbf{b}}_{p(1)}+\boldsymbol{\tau}-\tilde{\mathbf{a}}_1\right|-\left|\tilde{\mathbf{b}}_{p(j)}+\tilde{\mathbf{t}}_j'-\tilde{\mathbf{b}}_{p(1)}-\tilde{\mathbf{a}}_j+\tilde{\mathbf{a}}_1\right|\bigg|^\ell\right)
	      \\&=\min_{x\ge 0}\left(\theta_1 x^\ell + \theta_j\lvert x-y\rvert^\ell\right)
	      =\begin{cases}
		      y\min\{\theta_1,\theta_j\},&\text{if}~\ell=1\\
		      y^\ell \left(\theta_1^{-\ell+1}+\theta_j^{-\ell+1}\right)^{-\ell+1},&\text{if}~\ell>1
	      \end{cases}\label{ymint}
	.\end{align}
	\end{widetext}
	Since $y$ is unbounded, Eq.~(\ref{ymint}) can be arbitrarily large, which contradicts Eq.~(\ref{djp0t}).
\end{proof}
\begin{thm}\label{thm:finite-csm}
	The number of noncongruent CSMs in $\operatorname{CSM}^\star (\mathcal{A},\mathcal{B})$ is finite regardless of $\mu_\text{max}$, $w_\text{max}$, and $d_\text{max}$.
\end{thm}
\begin{proof}[\indent Proof]
	Replace the root node ``all CSMs'' in Fig.~\ref{fig:tree} with ``all noncongruent CSMs''.
	Theorems~\ref{thm:finite-slm} and \ref{thm:finite-shuffle} ensure that the breadth of the pruned tree is finite at each level.
\end{proof}
\section{Algorithms}\label{sec:algs}
To enumerate all noncongruent CSMs in $\operatorname{CSM}^\star (\mathcal{A},\mathcal{B})$, Theorem~\ref{thm:complete} allows us to enumerate all noncongruent SLMs in $\operatorname{SLM}^\star (\mathcal{A},\mathcal{B})$ first, and then all noncongruent CSMs within each SLM.
With Theorems~\ref{thm:imt} and \ref{thm:pct}, we only need to generate a complete subset of
\begin{multline}\label{imtmh}
	\operatorname{IMT}^\star (\mu)=\big\{(H_A,H_B,Q)\in\operatorname{IMT}(\mu)\,\big|
	\\w(C_BH_BQH_A^{-1}C_A^{-1})\le w_\text{max}\big\}
\end{multline}
for $\mu=1,\cdots,\mu_\text{max}$, and a complete subset of
\begin{multline}\label{pctla}
	\operatorname{PCT}^\star (\tilde L_A,\tilde L_B,S)=\big\{(p,\mathbf{t}_1,\cdots,\mathbf{t}_{\tilde Z})
	\\\in\operatorname{PCT}(\tilde L_A,\tilde L_B,S)\,\big|\,d(\mathcal{J}_{p,\mathbf{t}_1,\cdots,\mathbf{t}_{\tilde Z}})\le d_\text{max}\big\}
\end{multline}
for each $(\tilde L_A,\tilde L_B,S)\in\operatorname{SLM}^\star (\mathcal{A},\mathcal{B};\mu)$.
In this section, we will introduce Algorithm~\ref{alg:imt} for IMTs as well as Algorithms~\ref{alg:pct} and \ref{alg:best} for PCTs, and then apply them on some typical SSPTs.
These algorithms do not involve electronic structure or atomic force field calculations, i.e., they are purely geometric.
The pseudocodes are presented in an easy-to-understand form, but the actual implementation can be more efficient, e.g., all loops in Algorithm~\ref{alg:imt} can be parallelized.
\subsection{Enumerating sublattice matches}\label{ssec:imt}
Let's first consider how to produce the entire $\operatorname{IMT}^\star(\mu)$ rather than its complete subset.
It is easy to enumerate all possible values of $H_A$ and $H_B$, so the difficulty lies solely in computing all possible $Q$ for a given HNF pair $(H_A,H_B)$, i.e., producing $\{Q\,|\,(H_A,H_B,Q)\in\operatorname{IMT}^\star (\mu)\}$, which we denote by $\{Q\}$ for simplicity.
To this end, we first sample the 9D region
\begin{equation}
	\Omega_{S_0}=\left\{S_0\in\mathbb{R}^{3\times 3}\,\middle|\,w(S_0)\le w_\text{max}\right\}
,\end{equation}
producing a finite set consisting of \textit{trial} matrices, denoted as $\{S_0\}\subset\Omega_{S_0}$.
Then, for each trial matrix $S_0$, we compute the closest IMTs by (1) substituting $S_0$ into Eq.~(\ref{scbhb}), (2) solve for $Q$ for each pair $(H_A,H_B)$, and (3) round each matrix element of $Q$ to the nearest integer.
In other words, we apply the linear mapping
\begin{align}
	\phi_{H_A,H_B}\colon\mathbb{R}^{3\times 3}&\to\mathbb{R}^{3\times 3},
	\\S_0&\mapsto H_B^{-1}C_B^{-1}S_0C_AH_A,
\end{align}
to the set $\{S_0\}$.
As long as $\{S_0\}$ is sampled dense enough, its image under $\phi_{H_A,H_B}$ will intersect each $Q$'s ``preimage of rounding''
\begin{equation}\label{round}
	\operatorname{round}^{-1}(Q)=\left\{Q+X\,\middle|\,X\in\left(-\frac{1}{2},\frac{1}{2}\right)^{3\times 3}\right\}
,\end{equation}
a 9D hypercube centered at $Q$ [Fig.~\ref{fig:alg1}(b)].
In that case, we can extract $\{Q\}$ from $\operatorname{round}[\phi_{H_A,H_B}(\{S_0\})]$.
\begin{figure}[t!]
  \centering
  \includegraphics[width=\linewidth]{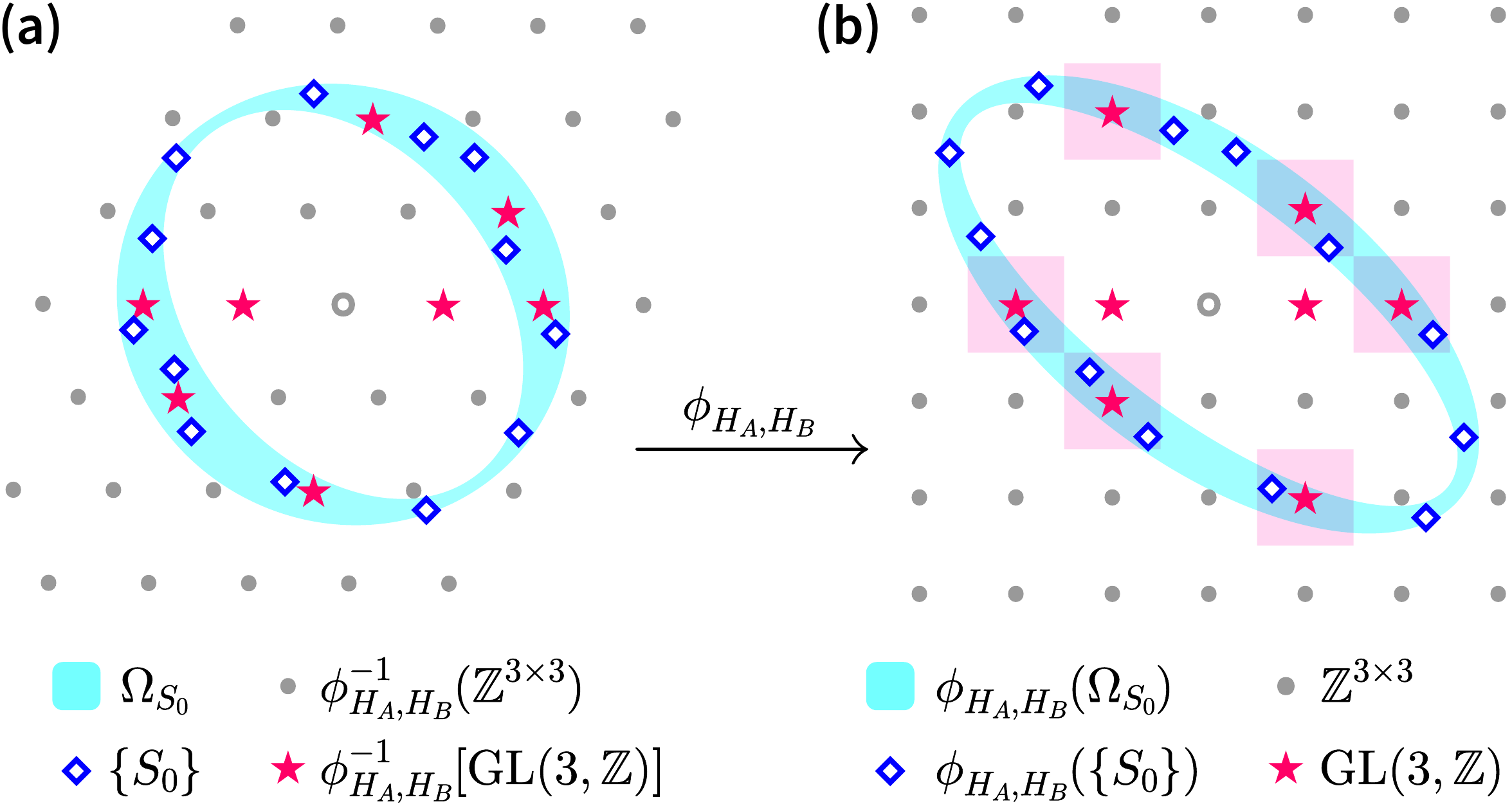}
  \caption{A schematic diagram showing how Algorithm~\ref{alg:imt} works, where 9D regions are drawn as 2D shapes for intuitive understanding.
	  (a) $\Omega_{S_0}$ (cyan) is a neighborhood of $\operatorname{SO}(3)$ determined solely by $Y^A_{ijkl}$, $Y^B_{ijkl}$ and $w_\text{max}$.
	  The trail matrices $\{S_0\}$ (blue diamonds) are generated in $\Omega_{S_0}$, approximating the exact deformation gradients (red stars).
	  (b) $\phi_{H_A,H_B}$ is applied to $\{S_0\}$, and then each $\phi_{H_A,H_B}(S_0)$ is rounded to the nearest integer matrix (gray dots).
	  Each element of $\{Q\}$ (red stars in cyan region) is produced if and only if a $S_0$ lies within its ``preimage of rounding'' (pink squares).
	  Currently, the algorithm should continue to generate a denser $\{S_0\}$ to produce the topmost $Q$ in the figure.
  }
	\label{fig:alg1}
\end{figure}
The only remaining problem is how to tell whether the set $\{S_0\}$ is dense \textit{enough} such that $\{Q\}$ is covered by its image under $\phi_{H_A,H_B}$.
Our approach is to generate $S_0$ one by one until $i_\text{th}$ consecutive iterations fail to produce any new IMT, as described in Algorithm~\ref{alg:imt}.
The value of $i_\text{th}$ can be derived by hypothesis testing.
Whenever the algorithm starts or a new IMT is produced, we initiate a test of the two opposing hypotheses, namely
\begin{align*}
	H_0\colon&\text{The current }\mathsf{myimts}\text{ is complete},\\
	H_1\colon&\text{The current }\mathsf{myimts}\text{ is incomplete}
.\end{align*}
Then, when we perform $i$ iterations without producing any new IMT, the conditional probability of this event under $H_0$ is $1$, whereas under $H_1$ it is only
\begin{equation}
	\Pr(i\,|\,H_1)\le (1-\eta)^i
,\end{equation}
where $\eta$ denotes the minimum probability of a given IMT being generated in a single iteration.
If one would reject $H_1$ when $\Pr(i \,|\, H_1) < \epsilon$, it is sufficient to set
\begin{equation}\label{ithln}
	i_\text{th} = \frac{\ln \epsilon}{\ln (1 - \eta)}.
\end{equation}
Note that the probability $\epsilon$ of incorrectly asserting completeness decreases exponentially with increasing $i_\text{th}$.
In this sense, we claim that Algorithm~\ref{alg:imt} is essentially an \textit{exhaustive} search algorithm.
\begin{algorithm}[t]
	\DontPrintSemicolon
	\caption{Listing all noncongruent IMTs}
	\label{alg:imt}
	\SetArgSty{}
	\SetKwData{MyIMTs}{myimts}
	\SetKwFunction{Standardize}{Standardize}
	\KwIn{$C_A$, $C_B$, $Z_A$, $Z_B$, $\mu$, $w_\text{max}$}
	\vspace{1pt}
	initialize $\MyIMTs\leftarrow\{\}$;~$i\leftarrow 0$\;
	\While{$i\le i_\text{th}$}{
		generate a random $S_0\in\Omega_{S_0}$\;
		\ForEach{$H_A\in\operatorname{HNF}\!\left(\frac{\mu\operatorname{lcm}(Z_A,Z_B)}{Z_A}\right)$}{
			\ForEach{$H_B\in\operatorname{HNF}\!\left(\frac{\mu\operatorname{lcm}(Z_A,Z_B)}{Z_B}\right)$}{
			$Q\leftarrow\operatorname{round}(H_B^{-1}C_B^{-1}S_0C_AH_A)$\;
			 $S\leftarrow C_BH_BQH_A^{-1}C_A^{-1}$\;
			 \If{$Q\in\operatorname{GL}(3,\mathbb{Z})$ {\bf and} $w(S)\le  w_\text{max}$}{
				 $x\leftarrow$ \Standardize{$H_A,H_B,Q$}\;
				 \If{$x\notin\MyIMTs$}{
				 add $x$ to \MyIMTs\;
				 $i\leftarrow 0$\;
			 }
				 \lElse{$i\leftarrow i+1$}
			}
		}
		}
	}
	\KwRet{\MyIMTs}
\end{algorithm}
Recall that our goal is to list all noncongruent elements in $\operatorname{IMT}^\star(\mu)$.
We can treat each IMT as its congruence class---whenever an IMT is generated, if it is congruent to some element already in $\mathsf{myimts}$, it is considered not new and $i$ increases.
However, a better way might be using a \texttt{Standardize()} function to convert each IMT into the representative of its congruence class (e.g., the lexicographically smallest one), as shown in the pseudocode.
Either way will replace $1-\eta$ in Eq.~(\ref{ithln}) with the minimum probability that no IMT from a given congruence class is generated in a single iteration.
A conservative estimate of $i_\text{th}$ is provided in Appendix~\ref{append:alg}, while in practice, $i_\text{th}$ can be dynamically adjusted to accelerate the algorithm.
\subsection{Enumerating shuffles}\label{ssec:pct}
To list all noncongruent PCTs, Algorithm~\ref{alg:pct} involves three functions: (1) \texttt{OptimizePCT()} for finding the PCT with the minimum $d$, (2) \texttt{Fill()} for generating all noncongruent PCTs with a given $p$, and (3) \texttt{Split()} for excluding previously enumerated values of $p$.
All we need to do is to iteratively invoke them, as shown in Fig.~\ref{fig:alg2}.
Next, we will introduce the specific algorithms for implementing these three functions separately.
\begin{algorithm}[t]
 	\DontPrintSemicolon
	\caption{Listing all noncongruent PCTs}
 	\label{alg:pct}
 	\SetArgSty{}
 	\SetKwData{MyPCTs}{mypcts}
	\SetKwData{MyConstraints}{myconstraints}
	\SetKwData{Constraint}{constraint}
	\SetKwData{NewConstraint}{newcons}
	\SetKwFunction{OptimizePCT}{OptimizePCT}
	\SetKwFunction{Fill}{Fill}
	\SetKwFunction{Split}{Split}
	\KwIn{$\mathcal{A}$, $\mathcal{B}$, $(\tilde L_A,\tilde L_B,S)$, $d_\text{max}$}
	\vspace{1pt}
	initialize $\MyPCTs\leftarrow\{\}$;~~$\MyConstraints\leftarrow$ empty queue\;
	enqueue $\{\}$ to \MyConstraints\;
	\While{\MyConstraints is not empty}{
		dequeue \Constraint from \MyConstraints\;
		$d,(p,\mathbf{t}_1,\cdots,\mathbf{t}_{\tilde Z})\leftarrow$ \OptimizePCT{\Constraint}\;
		\If{$d\le d_\text{max}$}{
			\ForEach{$(\mathbf{t}_2',\cdots,\mathbf{t}_{\tilde Z}')\in$ \Fill{$p$}}{
				add $(p,\mathbf{0},\mathbf{t}_2',\cdots,\mathbf{t}_{\tilde Z}')$ to \MyPCTs\;
			}
			\ForEach{$\NewConstraint\in$ \Split{$\Constraint, p$}}{
				enqueue \NewConstraint to \MyConstraints\;
			}
		}
	}
	\KwRet{\MyPCTs}
\end{algorithm}
\begin{figure}[b]
  \centering
  \includegraphics[width=0.8\linewidth]{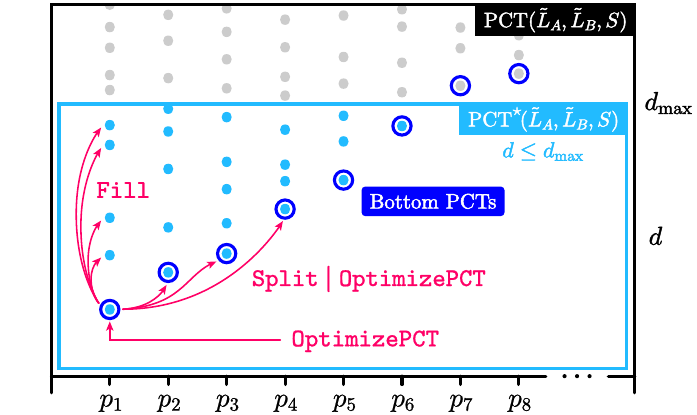}
  \caption{A schematic diagram showing how Algorithm~\ref{alg:pct} works, where each point represents a PCT.
	  First, \texttt{OptimizePCT()} is invoked to obtain the PCT with the smallest $d$.
	  Then, \texttt{Fill()} generates all PCTs that share the same permutation with it, while \texttt{Split()} produces a set of constraints to disable this permutation in subsequent computations.
	  By passing these constraints as parameters to \texttt{OptimizePCT()}, other bottom PCTs (dots in blue circles) can be generated.
  }
	\label{fig:alg2}
\end{figure}
\begin{algorithm*}[t]
 	\DontPrintSemicolon
	\caption{Optimizing PCT under a given constraint (producing the PCT with least shuffle distance)}
 	\label{alg:best}
 	\SetArgSty{}
	\SetKwProg{Def}{def}{}{}
	\SetKwData{Constraint}{constraint}
	\SetKwData{ReturnPCT}{returnpct}
	\SetKwFunction{Motifs}{Motifs}
	\SetKwFunction{ShuffleDistance}{ShuffleDistance}
	\SetKwFunction{BestAssignment}{BestAssignment}
	\KwIn{\Constraint, and all inputs of Algorithm~\ref{alg:pct}}
	\vspace{1pt}
	$\mathbf{a}_1^{(1)},\cdots,\mathbf{a}_{\tilde Z_1}^{(1)},\mathbf{a}_1^{(2)},\cdots,\mathbf{a}_{\tilde Z_2}^{(2)},\cdots,\mathbf{a}_1^{(n)},\cdots,\mathbf{a}_{\tilde Z_n}^{(n)}, \mathbf{b}_1^{(1)},\cdots,\mathbf{b}_{\tilde Z_1}^{(1)},\mathbf{b}_1^{(2)},\cdots,\mathbf{b}_{\tilde Z_2}^{(2)},\cdots,\mathbf{b}_1^{(n)},\cdots,\mathbf{b}_{\tilde Z_n}^{(n)}\leftarrow$ \Motifs{$\tilde L_A,\tilde L_B,S$}\;
	\Def{\ShuffleDistance{$\boldsymbol{\tau},\ReturnPCT=\mathit{false}$}}{
		\For{$\alpha=1$ \KwTo $n$}{
			\For{$i=1$ \KwTo $\tilde Z_\alpha$}{
				\For{$j=1$ \KwTo $\tilde Z_\alpha$}{
					$\mathbf{t}_{ij}^{(\alpha)}\leftarrow \operatorname{argmin}_{\mathbf{t}\in\tilde L_B}\left\lvert \tilde{\mathbf{b}}_{j}^{(\alpha)}+\tilde{\mathbf{t}}+\boldsymbol{\tau}-\tilde{\mathbf{a}}_{i}^{(\alpha)}\right\rvert^\ell$;\qquad $D_{ij}^{(\alpha)}\leftarrow \left\lvert \tilde{\mathbf{b}}_{j}^{(\alpha)}+\tilde{\mathbf{t}}_{ij}^{(\alpha)}+\boldsymbol{\tau}-\tilde{\mathbf{a}}_{i}^{(\alpha)}\right\rvert^\ell$\;
				}
			}
			$p^{(\alpha)}\leftarrow$ \BestAssignment{$D^{(\alpha)},\Constraint$};\qquad \lFor{$i=1$ \KwTo $\tilde Z_\alpha$}{$\mathbf{t}_{i}^{(\alpha)}\leftarrow \mathbf{t}_{i,p^{(\alpha)}(i)}^{(\alpha)}$}
		}
		$d\leftarrow \left(\sum_{\alpha=1}^n\sum_{i=1}^{\tilde Z_\alpha}\theta_i^{(\alpha)} D_{i,p^{(\alpha)}(i)}^{(\alpha)}\right)^{1 / \ell}$\;
		\lIf{\ReturnPCT}{\KwRet $d,\Big((p^{(1)}|p^{(2)}|\cdots|p^{(n)}),\mathbf{t}_{1}^{(1)},\cdots,\mathbf{t}_{\tilde Z_1}^{(1)},\mathbf{t}_{1}^{(2)},\cdots,\mathbf{t}_{\tilde Z_2}^{(2)},\cdots,\mathbf{t}_{1}^{(n)},\cdots,\mathbf{t}_{\tilde Z_n}^{(n)}\Big)$}
		\lElse{\KwRet $d$}
	}
	$\boldsymbol{\tau}^*\leftarrow \operatorname{argmin}_{\boldsymbol{\tau}\in \mathbb{R}^3}$ \ShuffleDistance{$\boldsymbol{\tau}$};\qquad \KwRet \ShuffleDistance{$\boldsymbol{\tau}^*, \ReturnPCT=\mathit{true}$}
\end{algorithm*}
We use \texttt{OptimizePCT(\textsf{constraint})} to generate the PCT with the smallest $d$ under the constraints specified by \textsf{constraint}.
In each iteration step, the overall translation $\boldsymbol{\tau}$ is fixed, and the best translation $\mathbf{t}_{ij}\in\tilde L_B$ for each atom pair $(\mathbf{a}_i, \mathbf{b}_j)$ is computed.
Then, the best permutation $p$ under this $\boldsymbol{\tau}$ is determined using \texttt{BestAssignment()} (e.g., the Jonker-Volgenant algorithm~\cite{jonker1987shortest,crouse2016implementing}), so that $\mathbf{t}_1,\cdots,\mathbf{t}_{\tilde Z}$ and
\begin{equation}
	\hat{d}_0(\boldsymbol{\tau})=\min_{p,\mathbf{t}_1,\cdots,\mathbf{t}_{\tilde Z}}\hat{d}[(+\boldsymbol{\tau})\circ\mathcal{J}_{p,\mathbf{t}_1,\cdots,\mathbf{t}_{\tilde Z}}]
\end{equation}
are also obtained.
This process is iterated with different $\boldsymbol{\tau}$ values until the global minimum of $\hat{d}_0(\boldsymbol{\tau})$ is found, where strategies such as basin-hopping are employed~\cite{wales1997global}.
The final iteration uses $\boldsymbol{\tau}^*=\operatorname{argmin}_{\boldsymbol{\tau}}\hat{d}_0(\boldsymbol{\tau})$, at which point the PCT given by \texttt{BestAssignment()} is
\begin{align}
	&\quad\,\mathop{\operatorname{argmin}}_{p,\mathbf{t}_1,\cdots,\mathbf{t}_{\tilde Z}}\hat{d}[(+\boldsymbol{\tau}^*)\circ\mathcal{J}_{p,\mathbf{t}_1,\cdots,\mathbf{t}_{\tilde Z}}]\\
	&=\mathop{\operatorname{argmin}}_{p,\mathbf{t}_1,\cdots,\mathbf{t}_{\tilde Z}}\min_{\boldsymbol{\tau}\in\mathbb{R}^3}\hat{d}[(+\boldsymbol{\tau})\circ\mathcal{J}_{p,\mathbf{t}_1,\cdots,\mathbf{t}_{\tilde Z}}]\\
	&=\mathop{\operatorname{argmin}}_{p,\mathbf{t}_1,\cdots,\mathbf{t}_{\tilde Z}}d(\mathcal{J}_{p,\mathbf{t}_1,\cdots,\mathbf{t}_{\tilde Z}})
,\end{align}
which is exactly what we need.
To formalize this idea and handle atomic species, the pseodocode is presented in Algorithm~\ref{alg:best}.
The function \texttt{Motifs($\tilde L_A, \tilde L_B, S$)} is a mapping like Eq.~(\ref{lalbs}) but also sorts the atoms by atomic species, where $n$ denotes the number of atomic species in the initial and final structures, $\tilde Z_\alpha$ denotes the number of atoms with the $\alpha$-th species, and other variables related to the $\alpha$-th atomic species are denoted with superscripts as $\mathbf{a}_i^{(\alpha)}$.
The function \texttt{BestAssignment($D, \mathsf{constraint}$)} returns a permutation $p$ that minimizes $\sum_{i}D_{i,p(i)}$, where $D$ is the cost matrix and \textsf{constraint} is a set of the form
\begin{equation}
	\left\{(i,j),(i',j'),\cdots;\overline{(k,l)},\overline{(k',l')},\cdots\right\}
,\end{equation}
where $(i,j)$ means enforcing $p(i)=j$, and $\overline{(k,l)}$ means requiring $p(k)\neq l$~\cite{murty1968algorithm}.
The correspondence of each atomic species is optimized separately and then combined into a single PCT, where $(p^{(1)}|p^{(2)}|\cdots|p^{(n)})\in\operatorname{Sym}(\tilde Z)$ represents performing $n$ permutations separately within $n$ atomic species.
We say that $(p,\mathbf{t}_1,\cdots,\mathbf{t}_{\tilde Z})$ is a bottom PCT if its $d$ is the smallest among all PCTs with the same permutation.
Whenever we obtain a bottom PCT, we invoke \texttt{Fill($p$)} to list all PCTs of the form $(p,\mathbf{0},\mathbf{t}_2',\cdots,\mathbf{t}_{\tilde Z}')$ that satisfy Eq.~\ref{djp0t}.
Here, we have let $\mathbf{t}_1'=\mathbf{0}$ to make the number of PCTs finite without loss of completeness (see the proof of Lemma~\ref{thm:finite-shuffle}).
Since $\mathbf{v}_i=\tilde{\mathbf{b}}_{p(i)}-\tilde{\mathbf{a}}_i$ is constant during each invocation of \texttt{Fill()}, one may simply use a flood fill algorithm to produce
\begin{multline}
	\big\{(\mathbf{t}_2',\cdots,\mathbf{t}_{\tilde Z}')\in\tilde L_B^{\tilde Z-1}\big|
	\\\min_{\boldsymbol{\tau}\in\mathbb{R}^3}\sum_{i=1}^{\tilde Z}\theta_i\lvert \mathbf{v}_i+\tilde{\mathbf{t}}_i+\boldsymbol{\tau}\rvert^\ell\le d_\text{max}^\ell \big\}
.\end{multline}
When $\ell=2$, Lemma~\ref{lem:rmsd} implies that $(\mathbf{t}_2,\cdots,\mathbf{t}_{\tilde Z}')$ takes values in a known hyperellipsoid so that one may employ a more efficient algorithm.
The last function \texttt{Split($\mathsf{constraint}, p$)} excludes the solution $p$ by splitting the assignment problem that satisfies \textsf{constraint} into several subproblems, which was originally proposed in Ref.~\citenum{murty1968algorithm}. It incorporates ``excluding $p$'' into \textsf{constraint} by returning $\tilde Z-1$ mutually disjoint sets
\begin{align}\label{const}
	\mathsf{constraint}&\cup \Big\{\overline{\big(1,p(1)\big)}\Big\},\\
	\mathsf{constraint}&\cup \Big\{\big(1,p(1)\big);\overline{\big(2,p(2)\big)}\Big\},\\
	\mathsf{constraint}&\cup \Big\{\big(1,p(1)\big),\big(2,p(2)\big);\overline{\big(3,p(3)\big)}\Big\},\\
	\begin{split}
		&\cdots\\
	\mathsf{constraint}&\cup \Big\{\big(1,p(1)\big),\cdots,\big(\tilde Z-2,p(\tilde Z-2)\big);
			 \\&\qquad\overline{\big(\tilde Z-1,p(\tilde Z-1)\big)}\Big\}.
	\end{split}
\end{align}
This partitions all feasible solutions that satisfy \textsf{constraint} and are not $p$.
If a set contains incompatible requirements, such as $(i, j)$ and $\overline{(i, j)}$, it will be discarded. 
In Algorithm~\ref{alg:pct}, whenever we obtain a bottom PCT, at most $\tilde Z-1$ sets like Eq.~(\ref{const}) are enqueued to \textsf{myconstraints}, each of which can yield a different bottom PCT via \texttt{OptimizePCT()}.
It should be noted that \textsf{mypcts} produced by Algorithm~\ref{alg:pct} is overcomplete---its elements are not noncongruent but rather appear in equivalence classes.
This is because for any $\mathbf{t}_A \in L_A$ and $\mathbf{t}_B \in L_B$, we have $(+\mathbf{t}_B)\circ\mathcal{J}\circ(-\mathbf{t}_A)$ congruent to $\mathcal{J}$, and their PCTs generally have different permutations unless $\mathbf{t}_B-S\mathbf{t}_A\in\tilde L_B$.
One may eliminate this redundancy by additionally invoking \texttt{Split()} to all ``congruent'' permutations, which is omitted here for brevity.
Another point to note is that a PCT generated by Algorithm~\ref{alg:pct} or \ref{alg:best} may represent a CSM whose SLM is not $(\tilde{L}_A, \tilde{L}_B, S)$, but rather $(\tilde{L}_A', \tilde{L}_B', S)$, where $\tilde{L}_A' \supset \tilde{L}_A$ and $\tilde{L}_B' \supset \tilde{L}_B$ are finer lattices.
This can only happen when $S\tilde L_A=\tilde L_B$ is a \textit{proper} subset of $SL_A\cap L_B$, and is a direct consequence of the fact that the PCT representation does not forbid shuffles with shuffle lattices finer than $\tilde L_B$.
However, this feature is useful in practical applications (see Section~\ref{ssec:or}) and does not affect the completeness of the enumeration.
\section{Application}\label{sec:app}
In this section, we demonstrate the applications of Algorithms~\ref{alg:imt}--\ref{alg:best} through case studies of the B1--B2 and graphite-to-diamond transitions.
For implementation details, see our Python package \textsc{crystmatch}~\cite{code}, which utilizes \textsc{spglib} for symmetry detection~\cite{togo2024spglib}.
We will see that the enumeration range easily covers and goes beyond previously proposed SSPT mechanisms.
The resulting database can be further analyzed from multiple perspectives, or used as inputs to NEB-like methods.
%
%Previously, a very coarse version of Algorithm~\ref{alg:imt} was employed for orientation relationship (OR) analysis in Ref.~\citenum{wang2024crystal}, which identified CSMs capable of producing the Kurdjumov-Sachs and Nishiyama-Wassermann ORs in the martensitic transformation of steel without requiring any rotation, as shown in Table~\ref{tab:mech}.
%
%Here, we adopt estimated strain energy $w$ instead of root-mean-square strain (RMSS) as the pruning criterion, and demonstrate the utility of Algorithms~\ref{alg:pct} and \ref{alg:best} in a compound system rather than an elemental one.
%
%We then compare the results under different pressures and introduce the concept of the lattice mismatch spectrum.
%
%The sensitivity to the norm type $\ell$ and the weights $\{ \theta_i \}$ is also discussed.
%
\subsection{Orientation relationship analysis}\label{ssec:or}
As one of the simplest SSPTs in compounds, the B1--B2 transition is similar to the FCC-to-BCC transition, but involves two atomic species.
Multiple distinct mechanisms---i.e., CSMs---have been proposed~\cite{stokes2004mechanisms}, including the purely distortive ($d=0$) Buerger one~\cite{buerger1951phase}, the Watanabe-Tokonami-Morimoto (WTM) one which can explain the observed orientation relationship (OR)~\cite{watanabe1977transition}, and the Tol{\' e}dano-Knorr-Ehm-Depmeier (TKED) one which is supported by MetaD simulations involving ${\sim}10^4$ atoms~\cite{toledano2003phenomenological,badin2021nucleating}.
Recently, the Therrien-Graf-Stevanovi{\' c} (TGS) mechanism was identified using the \textsc{p2ptrans} method~\cite{therrien2020minimization}, which has a significantly lower strain.
These CSMs and their evidence are summarized in Table~\ref{tab:mech}.
A CSM can be decomposed into an SLM and a PCT, where the former consists of a deformation gradient $S$ and a shuffle lattice.
Among these components, the $S$ alone determines the OR and the habit plane, and can thus be compared directly with experimental observations.
In contrast, the shuffle lattice and the PCT describe microscopic details that are, currently, accessible only through simulations.
In this sense, listing all possible $S$ is of broader interest, which can be achieved using Algorithm~\ref{alg:imt}.
To further characterize $S$ from a CSM perspective, we used Algorithm~\ref{alg:best} to compute the \textit{minimal multiplicity} and \textit{minimal shuffle distance} among all CSMs with deformation gradient $S$, denoted by $\mu_0(S)$ and $d_0(S)$, respectively.
For simplicity, we let $\ell=2$ and assign equal shuffle-distance weights to Cs and Cl atoms, i.e., set $\theta_i$ as a constant.
A more physical choice of weights will be discussed in Section~\ref{ssec:exhaust}.
\begin{figure}[t!]
  \centering
  \includegraphics[width=\linewidth]{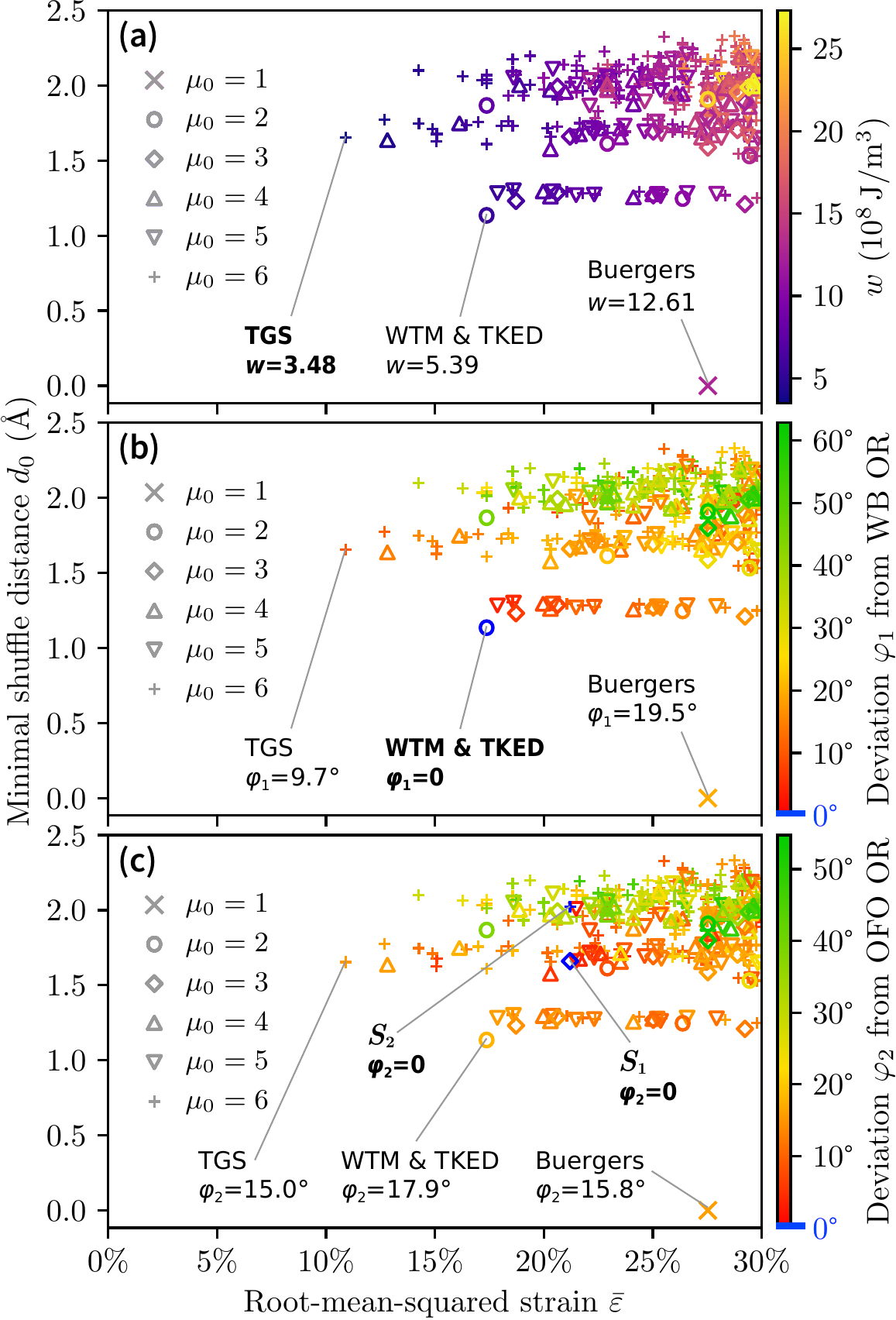}
  \caption{
All deformation gradients $S$ with $\mu_0\le 6$ and $\bar\varepsilon\le 30\%$.
Each point represents a unique $S$, whose marker shape indicates its $\mu_0$.
(a) Each $S$ is colored according to its estimated strain energy density $w(S)$.
Those $S$ of previously proposed mechanisms are pointed out by gray solid lines, where the WTM one and TKED one share the same $S$, and the TGS one has the lowest $w$ among all $S$ enumerated.
(b)(c) Each $S$ is colored according to its deviation angles $\varphi_1$ from the WB OR and $\varphi_2$ from the OFO OR.
Those $S$ with strictly $\varphi_1=0$ or $\varphi_2=0$ are highlighted in blue.}
	\label{fig:rmss}
\end{figure}
\begin{figure}[t!]
  \centering
  \includegraphics[width=\linewidth]{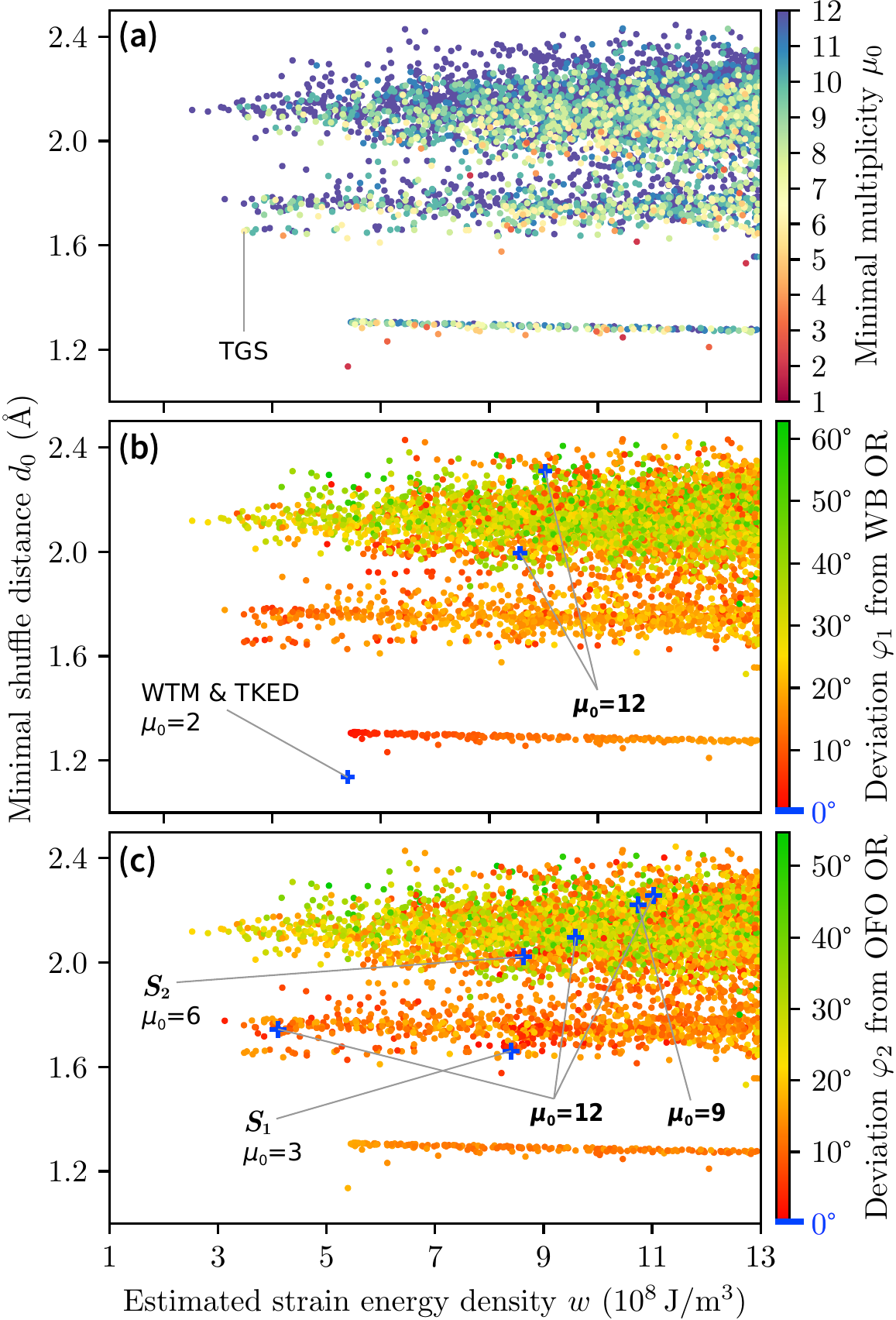}
  \caption{All deformation gradients $S$ with $\mu_0\le 12$ and $w\le 13\times 10^8\,\text{J}/\text{m}^3$.
	Only the Buerger mechanism is omitted, as it has already been shown in Fig.~\ref{fig:rmss}.	  
  (a) Each $S$ is colored according to $\mu_0$, the minimal multiplicity required for an SLM to have $S$.
  There are several $S$ with $w$ less than the TGS mechanism, but also higher $d_0$ and $\mu_0$.
  (b)(c) Each $S$ is colored according to the deviation angles.
  Blue crosses highlight those $S$ with $\varphi_1=0$ or $\varphi_2=0$.}
  \label{fig:b1b2}
\end{figure}
As a preliminary investigation, we enumerated all SLMs for the B1--B2 transition with multiplicity $\mu \leq 6$ (i.e., period $\tilde Z\le 12$) and RMSS
\begin{equation}
	\bar\varepsilon (S) = \sqrt{\frac{1}{3}\sum_{i=1}^3 [\sigma_i(S)-1]^2} 
\end{equation}
less than 30\%.
This was achieved by replacing $\Omega_{S_0}$ in Algorithm~\ref{alg:imt} with $\{ S\in\mathbb{R}^{3\times 3} \,|\, \bar\varepsilon(S) \le 0.3 \}$.
From the enumeration result, we identified the SLMs of all established CSMs above, as shown in Fig.~\ref{fig:rmss}(a).
A total of 788 noncongruent SLMs with 393 distinct deformation gradients are produced, taking 1\,min\,11\,s on a single core of a Mac M1 chip.
The subsequent $d_0$ calculation took only 13\,s, producing 393 CSMs with the smallest $d$ among all CSMs sharing their respective deformation gradients, i.e., satisfying
\begin{equation}
	d(\mathcal{J})=d_0(S)
.\end{equation}
We refer to a CSM with this property as a \textit{representative} CSM, whose PCT is necessarily a bottom PCT, though the converse does not hold.
\begin{figure}[b!]
  \centering
  \includegraphics[width=\linewidth]{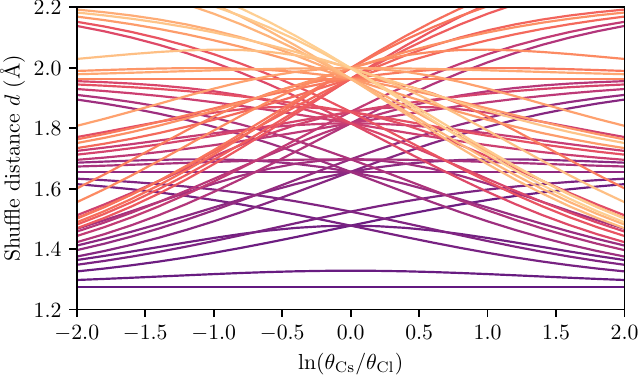}
  \caption{
How the shuffle distance $d$ depends on the weighting scheme, i.e., the ratio $\theta_\text{Cs} / \theta_\text{Cl}$ for SSPTs in CsCl.
The figure shows 51 CSMs with the same SLM ($\mu=6$), whose $d$ has only 8 distinct values when $\theta_\text{Cs}=\theta_\text{Cl}$.
Colors are used to distinguish between different CSMs.
}
  \label{fig:sensitivity}
\end{figure}
It should be noted that all the 12 mechanisms in Table~\ref{tab:mech} are representative CSMs---except for the TKED one.
This means that Algorithms~\ref{alg:imt} and \ref{alg:best} are already capable of reproducing most known CSMs, while Algorithm~\ref{alg:pct} serves as a last resort for exhaustive enumeration.
One could certainly apply Algorithm~\ref{alg:pct} to every SLM, but a more efficient approach is to filter them using experimentally observed ORs.
As detailed in Ref.~\citenum{stokes2004mechanisms}, the most commonly observed ORs in the B1--B2 transition are the Watanabe-Blaschko (WB) OR
\begin{equation}
	[001][110]_{\text{B1}}\parallel [110][001]_{\text{B2}} 
,\end{equation}
and the Okai-Fujiwara-Onodera (OFO) OR
\begin{equation}
	[100][111]_{\text{B1}}\parallel [111][100]_{\text{B2}}
.\end{equation}
For both ORs, we computed the respective deviation angles $\varphi_1$ and $\varphi_2$ of each deformation gradient, defined as the minimal rotation required to produce the OR, as shown in Fig.~\ref{fig:rmss}(b)(c).
We found that the deformation gradient of the WTM and TKED mechanisms is the only one that can produce the WB OR without rotation, which we denote by $S_\text{W}$.
On the other hand, the OFO OR can only be produced without rotation by two previously unreported deformation gradients, denoted as $S_1$ and $S_2$.
At this point, we can see that \textsc{crystmatch} provides a systematic way to infer the deformation gradient---and thus the SLM---based on experimental observations.
One naturally wonders whether, as $\mu_\text{max}$ increases, there exists a deformation gradient $S$ with lower strain energy $w$ than that of the TGS mechanism, or other $S$ that can reproduce either OR without rotation.
The answer to both questions is affirmative as long as we set $\mu_{\text{max}} = 12$, as shown in Fig.~\ref{fig:b1b2}.
Since the Buerger mechanism has $d = 0$, we need not consider mechanisms with higher strain energy than it.
Accordingly, we set $w_{\text{max}}=13 \times 10^{8} \,\mathrm{J}/\mathrm{m}^3$ and performed a more comprehensive enumeration that took about 2 CPU hours, obtaining 9987 noncongruent SLMs and 5432 distinct deformation gradients.
Despite the total numbers increasing by more than an order of magnitude, only a handful of new cases emerged with $w$ lower than that of the TGS mechanism, $\varphi_{1} = 0$, or $\varphi_{2} = 0$.
\begin{figure*}[t!]
  \centering
  \includegraphics[width=\textwidth]{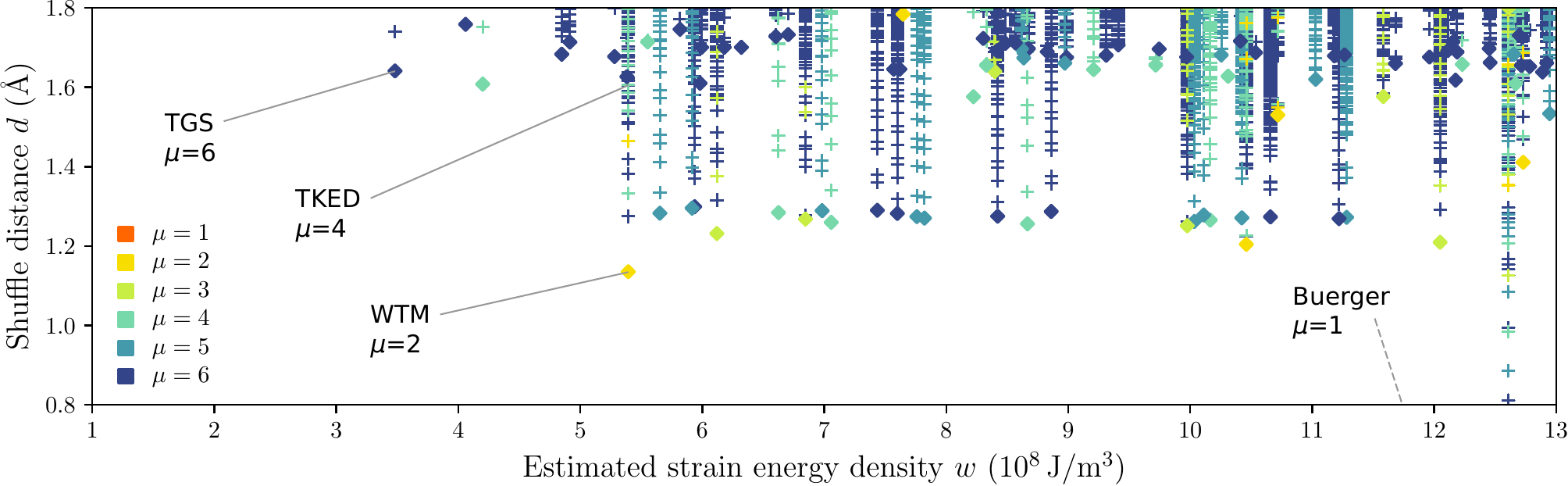}
  \caption{
  All CSMs with $\mu\le 6$, $w\le 13\times 10^8\,\text{J}/\text{m}^3$, and $d\le 1.8\,\text{\AA}$.
  Only the Buerger CSM is ommitted, as it has already been shown in Fig.~\ref{fig:rmss}.
  Diamonds and crosses indicate representative and nonrepresentative CSMs, respectively.
  Each CSM is colored according to its multiplicity.
}
  \label{fig:exhaust}
\end{figure*}
\subsection{Exhaustive enumeration}\label{ssec:exhaust}
Now we assign shuffle-distance weights to Cs and Cl atoms.
If equal weights are used, a large number of CSMs will exhibit identical $d$ values, as shown in Fig.~\ref{fig:sensitivity}.
A physically reasonable approach is to consider the initial velocity $v$ required to move from the initial state to the transition state on the PES within an average free time, which is proportional to the unweighted $d$.
Statistical mechanics tells us that heavier atoms are less likely to acquire velocities, since kinetic energy $T_i=\frac{1}{2}m_iv_i^2$.
Therefore, each atom's contribution to $d^2$ should be weighted as $\theta_i\propto m_i$, and $\ell = 2$ is thus an appropriate choice.
This leads to $\ln(\theta_\text{Cs} / \theta_\text{Cl})\approx 1.32$, which is sufficient to lift the degeneracy in $d$.
\begin{figure}[b!]
  \centering
  \includegraphics[width=\linewidth]{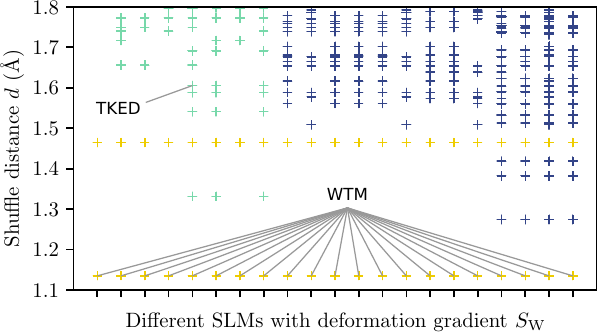}
  \caption{All CSMs with $\mu\le 6$ and deformation gradient $S_\text{W}$.
	  CSMs that can be produced by Algorithm~\ref{alg:pct} using the same SLM are drawn in the same horizontal coordinate.
  Yellow, green, and indigo represent $\mu=2,4,6$, respectively.
  Gray lines point out the CSMs of the WTM and TKED mechanisms.
  Note that CSMs with $\mu=2$ also acquire PCT representations in SLMs with $\mu=4,6$.}
  \label{fig:wtm}
\end{figure}
To investigate the distribution of nonrepresentative CSMs and to reproduce the TKED mechanism, Algorithm~\ref{alg:pct} is applied to each SLM with $\mu\le 6$ and $w\le 13\times 10^8\,\text{J}/\text{m}^3$, using $d_\text{max}=1.8\,\text{\AA}$.
This took about 3 hours, producing 320\,757 CSMs, as shown in Fig.~\ref{fig:exhaust}.
However, if one is only interested in the CSMs that exactly conform to the WB OR, it is sufficient to apply Algorithm~\ref{alg:pct} only to those SLMs with $S_\text{W}$, taking merely 7\,min\,39\,s.
Figure~\ref{fig:wtm} shows all CSMs with $S_\text{W}$, from which one can see that while $S_\text{W}$ belongs to the SLM of the WTM mechanism ($\mu=2$), it also belongs to 7 noncongruent SLMs with $\mu=4$ and 13 noncongruent SLMs with $\mu=6$.
This is because for a given $(H_A,H_B,Q)\in\operatorname{IMT}(\mu)$, any $M\in\mathbb{Z}^{3\times 3}$ with $\det M>0$ induces an $(H_A',H_B',Q')\in\operatorname{IMT}(\mu\det M)$ with the same deformation gradient, where
\begin{align}
	H_A'&=\operatorname{hnf}(H_AM),\label{hahnf}
	\\H_B'&=\operatorname{hnf}(H_BQM),\label{hbhnf}
	\\Q'&=H_B'^{-1}H_BQH_A^{-1}H_A'\label{qhb-1}
.\end{align}
Overall, the total number of CSMs increases rapidly with both $\mu$ and $d_\text{max}$.
\begin{figure*}[t!]
  \centering
  \includegraphics[width=\textwidth]{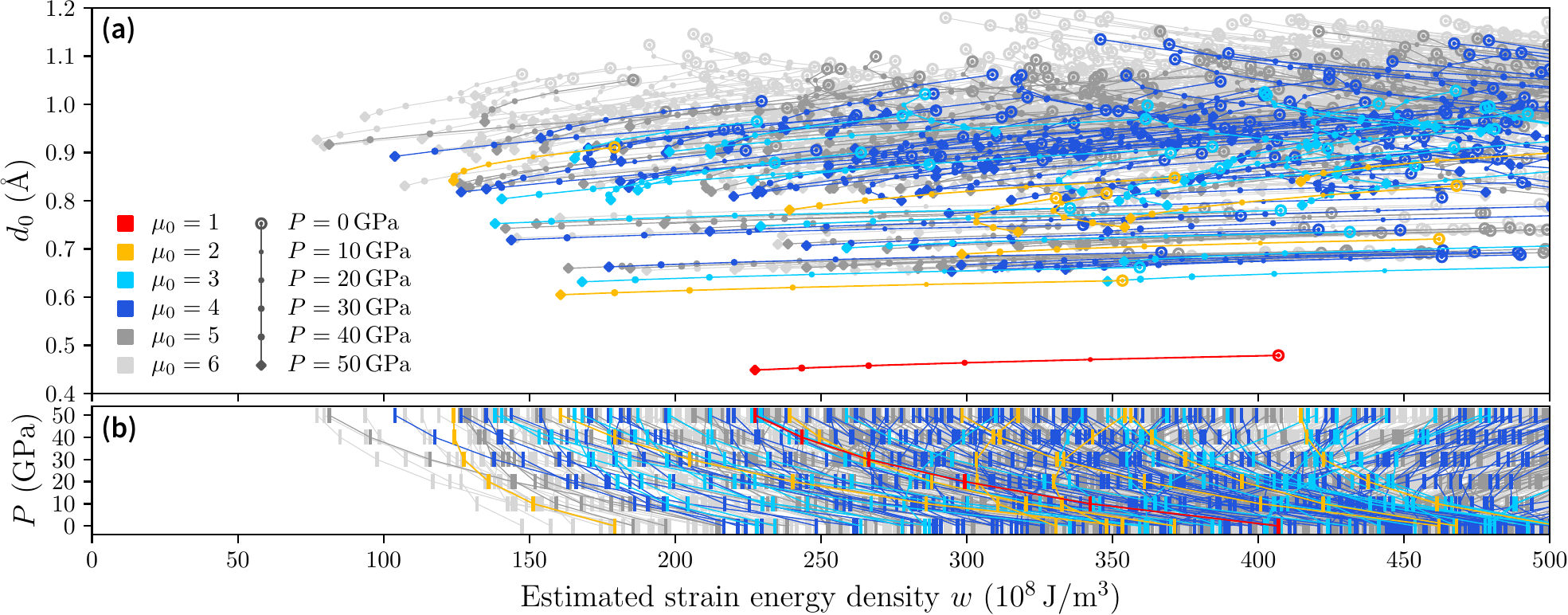}
  \caption{
  All deformation gradients $S$ with $\mu_0\le 6$ and $w\le 500\times 10^8\,\text{J}/\text{m}^3$.
  Colors are used to distinguish between different multiplicities.
  (a) The distribution of $w(S)$ and $d_0(S)$ at each pressure $P$.
  Solid lines connect points representing the same SLM under different $P$, whose marker shapes change with increasing pressure, as indicated in the legend at the lower left.
  (b) The lattice mismatch spectrum at each $P$.
  Solid lines show the diverse behaviors of $w$ as a function of $P$ for different SLMs.
}
  \label{fig:spec}
\end{figure*}
Note that $\mathcal{J}_{p,\mathbf{t}_1,\cdots,\mathbf{t}_{\tilde Z}}$ satisfying Eq.~(\ref{aitaj}) may have a period smaller than $\tilde Z$.
Consequently, the WTM CSM $\mathcal{J}_\text{W}$ is reproduced in every $\operatorname{PCT}^\star (\tilde L_A,\tilde L_B,S_\text{W})$ such that $\tilde L_B$ is a sublattice of the shuffle lattice of $\mathcal{J}_\text{W}\circ S_\text{W}^{-1}$; see Fig.~\ref{fig:wtm}.
In fact, all CSMs exhibit this repetitive property under the PCT representation.
Although one could easily eliminate it by modifying Algorithm~\ref{alg:pct} and Definition~\ref{def:pct}, doing so is unnecessary, since the current formalism ensures that $\mathrm{PCT}(\tilde L_A, \tilde L_B, S)$ precisely consists of all possible CSM that a molecular-dynamics simulation employing periodic boundary conditions $\tilde L_A$ (initial) and $\tilde L_B$ (final), undergoing the deformation $S$, may produce---even when the resulting CSM possesses a smaller intrinsic period.
\subsection{Effect of pressure on lattice mismatch}
The densities of two solid phases of the same material can differ significantly, making their lattice mismatch and strain energy inevitably large.
In such cases, pressure plays a decisive role in the occurrence of SSPTs.
For example, at 1500\,K, the graphite-diamond coexistence pressure estimated by the Berman-Simon line is $\sim$5\,GPa, while the SSPT is not observed until $\sim$12\,GPa~\cite{bundy1961diamond,irifune2003ultrahard}.
Previous studies also suggest that when the pressure is further increased to 20--50\,GPa, both the critical nucleus size and nucleation energy decrease, as shown in Fig.~2 of Ref.~\citenum{khaliullin2011nucleation}.
However, lattice mismatch cannot be quantified unless how lattices match is self-evident, or defined in advance---this is precisely what the concept of SLM provides.
By enumerating all possible deformation gradients, Algorithm~\ref{alg:imt} enables a comprehensive analysis of lattice mismatch.
We investigated the graphite-to-diamond transition under pressures $P=0,10,20,30,40,50\,\text{GPa}$.
The enumeration range was specified as $\mu_\text{max} = 6$ (i.e., $\tilde Z\le 24$) and $w_\text{max} = 5\times 10^{10}\,\text{J}/\text{m}^3$ to ensure inclusion of all established mechanisms, with the initial and final structures as well as elastic tensors calculated at each $P$.
The minimal multiplicity $\mu_0$ and minimal shuffle distance $d_0$ for each deformation gradient is also computed as we did in Section~\ref{ssec:or}.
For the case of $P = 0$, enumerating the SLMs and computing $d_0$ took 9\,min\,29\,s and 37\,s, respectively.
The joint distribution of $w$ and $d_0$, as well as the spectrum of $w$ are shown in Fig.~\ref{fig:spec}.
We can see that as pressure increases, the strain energy of those deformation gradients with small $d_0$ or small $w$ decreases significantly.
However, some other deformation gradients---particularly those with larger $\mu_0$ and $d_0$---exhibit opposite or even nonmonotonic trends.
This suggests that the actual SSPT mechanism may vary under different pressures.
Overall, the ``lattice mismatch spectrum'' in Fig.~\ref{fig:spec}(b) extends leftward with increasing pressure, consistent with previous experimental and computational studies.
\section{Summary}
We have introduced a cell-independent formalism that unequivocally defines the deformation gradient, SLM, and standard shuffle of a CSM, as well as the symmetry-induced congruence relations on $\operatorname{CSM}(\mathcal{A},\mathcal{B})$ and $\operatorname{SLM}(\mathcal{A},\mathcal{B})$.
This maps all possible CSMs of an SSPT onto a tree structure, which is finite up to congruence classes as long as the multiplicity, strain energy and shuffle distance are bounded.
Using the IMT and PCT representation, each CSM is encoded in a tuple of integer matrices with minimal redundancy.
Based on this, Algorithms~\ref{alg:imt}--\ref{alg:best} have been proposed for the exhaustive enumeration of CSMs and SLMs, unveiling novel SSPT mechanisms with unprecedentedly low strain or desired ORs.
They lay the groundwork for predicting CSM energy barriers via machine learning and can be seamlessly integrated with widely used methods with negligible computational cost.
One may also use them to study defect migration or heteroepitaxy~\cite{wei2021direct,li2022smallest,li2025silicon}.
All algorithms are implemented in the command-line tool \textsc{crystmatch}, which is open-sourced and freely available to the community~\cite{code}.
\appendix
\begin{acknowledgements}
We acknowledge helpful discussions with Xu-Yuan Zhou, Yuze Sun, Wei-Jian Jiang, Yi-Chi Zhang, Jia-Xi Zeng, Qi Liu, Yi Yang, Qing-Yang Zheng, and Jia-Cheng Gao.
We are supported by the National Science Foundation of China under Grants No.~123B2048, No.~12204015, No.~12234001, No.~12404257, No.~12474215, and No.~62321004, and the National Basic Research Program of China under Grants No.~2021YFA1400500 and No.~2022YFA1403500.
The computational resources were provided by the supercomputer center at Peking University, China.
\end{acknowledgements}
\section{Lemmas and proofs}\label{append:lemma}
\begin{proof}[\indent Proof of Lemma~\ref{lem:inducedshuffle}]
	$\mathbf{t}\in\mathbb{R}^3$ is a translation element of $\mathcal{J}$ if and only if Eq.~(\ref{ajbtt}) holds, which is equivalent to:
	\begin{equation}
		\begin{tikzcd}[column sep=scriptsize]
			\mathcal{A}+\mathbf{t}_A&&\mathcal{A}\arrow[ll, "+\mathbf{t}_A"']\arrow[r, "\mathcal{J}", rightsquigarrow]&\mathcal{B}\arrow[rr, "+\mathbf{t}_B"]&&\mathcal{B}+\mathbf{t}_B\\
			\mathcal{A}+\mathbf{t}_A&&\mathcal{A}\arrow[ll, "+\mathbf{t}_A"]\arrow[r, "\mathcal{J}"', rightsquigarrow]&\mathcal{B}\arrow[rr, "+\mathbf{t}_B"']&&\mathcal{B}+\mathbf{t}_B
		  \arrow[from=1-1, to=2-1, "+\mathbf{t}"']
		  \arrow[from=1-3, to=2-3, "+\mathbf{t}"']
		  \arrow[from=1-4, to=2-4, "+\mathbf{t}"]
		  \arrow[from=1-6, to=2-6, "+\mathbf{t}"]
	  \end{tikzcd}
	.\end{equation}
	The outer loop of the above diagram means that $\mathbf{t}$ is a translation element of $(+\mathbf{t}_B)\circ \mathcal{J}\circ(-\mathbf{t}_A)$.
	The above derivation is reversible, so that the two CSMs have the same translation elements.
\end{proof}
\begin{proof}[\indent Proof of Lemma~\ref{lem:basis}]
	We prove the following proposition for each positive integer $n$ by induction: If $L$ is an addition group on $\mathbb{R}^n$ with rank $n$ and
	\begin{equation}\label{ilinf}
		\lambda(L)=\inf_{\mathbf{t}\in L\setminus\{\mathbf{0}\}}\lvert \mathbf{t}\rvert > 0
	,\end{equation}
	then
	\begin{equation}\label{t1tnl}
	  \exists \mathbf{t}_1,\cdots,\mathbf{t}_n,\quad L=\left\{\sum_{i=1}^nk_i\mathbf{t}_i\,\middle|\,k_1,\cdots,k_n\in\mathbb{Z}\right\}
	.\end{equation}
	Since Eq.~(\ref{ilinf}) implies Eq.~(\ref{t1t2l}), Lemma~\ref{lem:basis} is nothing but a special case of this proposition for $n=3$.

	\textit{Base case. }For $n=1$, let $\mathbf{t}_1$ denote an element of $L$ with length $\lambda(L)$.
	Such an element must exist; otherwise, $L$ would have infinitely many elements with bounded lengths, and hence a cluster point (by the Bolzano-Weierstrass theorem), which contradicts Eq.~(\ref{t1t2l}).
	Since $L$ is an addition group, we have $\{k_1\mathbf{t}_1\,|\,k_1\in\mathbb{Z}\}\subset L$.
	There are no other elements in $L$; otherwise, the minimum distance between elements in $L$ would be smaller than $\lambda(L)$.

	\textit{Induction step. }Now we show that for every $m\ge 1$, if Eq.~(\ref{t1tnl}) holds for $n=m$, it also holds for $n=m+1$.
	Let $\mathbf{t}_1$ denote an element of $L\subset \mathbb{R}^{m+1}$ with length $\lambda(L)$.
	Any element in $\mathbb{R}^{m+1}$ can be uniquely decomposed as
	\begin{equation}\label{vat1v}
	  \mathbf{v}=c_1\mathbf{t}_1+\mathbf{v}'
	,\end{equation}
	where $c_1\in\mathbb{R}$ and $\mathbf{v}'\in\operatorname{span}_{\mathbb{R}}\{\mathbf{t}_1\}^\perp$ is orthogonal to $\mathbf{t}_1$.
	This defines a projection mapping
	\begin{align}
		\pi\colon \mathbb{R}^{m+1}&\to \operatorname{span}_{\mathbb{R}}\{\mathbf{t}_1\}^\perp,\\
		\mathbf{v}&\mapsto\mathbf{v}'
	.\end{align}
	Since $\pi$ is a linear mapping, the image of $L$ under this mapping, $\pi(L)$, is an additive group with rank $m$.
	In fact, we also have $\lambda(\pi(L)) > 0$; otherwise, $L$ would have infinitely many elements with bounded lengths, and hence a cluster point.
	Therefore, $\pi(L)$ satisfies the proposition for $n=m$, so there must exist $\mathbf{t}_2,\cdots,\mathbf{t}_{m+1}\in L$ whose images under $\pi$ generate
	\begin{equation}\label{pli2m}
		\pi(L)=\left\{\sum_{i=2}^{m+1}k_i\pi(\mathbf{t}_i)\,\middle|\,k_2,\cdots,k_{m+1}\in\mathbb{Z}\right\}
	.\end{equation}
	Also, since $L$ is an addition group, we have
	\begin{equation}
	  \left\{\sum_{i=1}^{m+1}k_i\mathbf{t}_i\,\middle|\,k_1,\cdots,k_{m+1}\in\mathbb{Z}\right\}\subset L
	,\end{equation}
	and we only need to show that there are no other elements in $L$.
	Note that any $\mathbf{t}\in L$ can be expressed as
	\begin{equation}\label{tc1t1}
	  \mathbf{t}=c_1\mathbf{t}_1+\sum_{i=2}^{m+1}c_i\mathbf{t}_i
	,\end{equation}
	for some $c_1,\cdots,c_{m+1}\in\mathbb{R}$.
	Applying $\pi$ to Eq.~(\ref{tc1t1}) and use Eq.~(\ref{pli2m}), we can see that $c_2,\cdots,c_{m+1}$ must be integers.
	This implies that
	\begin{equation}\label{c1c1t}
		\big(c_1-\lfloor c_1\rfloor\big)\mathbf{t}_1\in L
	,\end{equation}
	where $\lfloor c_1\rfloor$ denotes the greatest integer less than or equal to $c_1$.
	If $c_1\notin \mathbb{Z}$, the length of Eq.~(\ref{c1c1t}) would be nonzero and smaller than $\lambda(L)$, which contradicts Eq.~(\ref{ilinf}).
	At this point, we can see that all coefficients in Eq.~(\ref{tc1t1}) are integers, so that the proposition holds for $n=m+1$.
\end{proof}
\begin{proof}[\indent Proof of Lemma~\ref{lem:sublat}]
	If $\tilde C=CM$ for some $M\in\mathbb{Z}^{3\times 3}$, then $\tilde C(\mathbb{Z}^3)=C(M(\mathbb{Z}^3))$ is a sublattice of $C(\mathbb{Z}^3)$ since $M(\mathbb{Z}^3)$ is a sublattice of $\mathbb{Z}^3$.
	Conversely, if $\tilde C(\mathbb{Z}^3)$ is a sublattice of $C(\mathbb{Z}^3)$, there exist $\mathbf{k}_1,\mathbf{k}_2,\mathbf{k}_3\in\mathbb{Z}^3$ such that the $i$-th column of $\tilde C$ equals $C\mathbf{k}_i$, i.e., $\tilde C=C[\mathbf{k}_1,\mathbf{k}_2,\mathbf{k}_3]$.
	To see that the index of this sublattice is $\lvert\det M\rvert$, we only need to use the elementary divisor theorem~\cite{cohen2011course}, which states that the quotient group $C(\mathbb{Z}^3) / \tilde C(\mathbb{Z}^3)$ is isomorphic to a $\mathbb{Z}$-module of size $\lvert \det M\rvert $.
\end{proof}
\begin{proof}[\indent Proof of Lemma~\ref{lem:quotient}]
	Arbitrarily take $\mathbf{a}_1,\cdots,\mathbf{a}_{Z_A}\in\mathcal{A}$ that are $L_A$-inequivalent, and $\mathbf{t}_1,\cdots,\mathbf{t}_k\in L_A$ that are $\tilde L_A$-inequivalent.
	We have
	\begin{align}
	  \mathcal{A}
	  &=\bigsqcup_{i=1}^{Z_A} (\mathbf{a}_i+L_A)
	\\&=\bigsqcup_{i=1}^{Z_A} \left(\mathbf{a}_i+\bigsqcup_{j=1}^k(\mathbf{t}_j+\tilde L_A)\right)
	\\&=\bigsqcup_{i=1}^{Z_A}\bigsqcup_{j=1}^k\left((\mathbf{a}_i+\mathbf{t}_j)+\tilde L_A\right)
	,\end{align}
	where $\sqcup$ denotes the disjoint union of sets.
	Since each $(\mathbf{a}_i+\mathbf{t}_j)+\tilde L_A$ is a distinct $\tilde L_A$-equivalence class, we have $\lvert \mathcal{A}/\tilde L_A\rvert=kZ_A$.
\end{proof}
\begin{proof}[\indent Proof of Lemma~\ref{lem:redundancy}]
	Note that
	\begin{align}
		C(\mathbb{Z}^3)=C'(\mathbb{Z}^3)\quad
		&\iff\quad C^{-1}C'(\mathbb{Z}^3)=\mathbb{Z}^3\\
		&\iff\quad C^{-1}C'\in\operatorname{GL}(3,\mathbb{Z})
	,\end{align}
	where we have used Eq.~(\ref{qznzn}).
\end{proof}
\begin{proof}[\indent Proof of Lemma~\ref{lem:hnf}]
	The existence of $H$ and $Q$ as well as the uniqueness of $H$ are proved in Ref.~\citenum{cohen2011course}.
	To see that $Q$ is also unique, note that $Q=H^{-1}M$.
\end{proof}
\begin{proof}[\indent Proof of Lemma~\ref{lem:svd}]
	$S^{\text{T}}S$ is positive-definite, so it has three orthogonal eigenvectors $\mathbf{v}_1,\mathbf{v}_2,\mathbf{v}_3$ (arranged in right-handed order) with eigenvalues $\lambda_1\ge\lambda_2\ge\lambda_3>0$.
	Let $s_i=\sqrt{\lambda_i} $, $V=[\mathbf{v}_1,\mathbf{v}_2,\mathbf{v}_3]\in\operatorname{SO}(3)$, and $U=SV\Sigma^{-1}$.
	We have $U\in \operatorname{SO}(3)$ since $U^\text{T}U=I$ and $\det U=1$.
	Since $\lambda_i$ is uniquely determined by $S$, so is $s_i$.
\end{proof}
\begin{lem}\label{lem:rmsd}
	When $\ell = 2$, the right-hand side of Eq.~(\ref{djmin}) attains the minimum if and only if $\boldsymbol{\tau}=\overline{\tilde{\mathbf{a}}_i }-\overline{\tilde{\mathcal{J}}(\tilde{\mathbf{a}}_i )}$, where the overline denotes the mean value
	\begin{equation}
		\overline{\mathbf{v}_i}=\sum_{i=1}^{\tilde Z}\theta_i\mathbf{v}_i
	.\end{equation}
	Consequently, a CSM with PCT $(p,\mathbf{t}_1,\cdots,\mathbf{t}_{\tilde Z})$ has
	\begin{equation}\label{djbpi}
		d(\mathcal{J}_{p,\mathbf{t}_1,\cdots,\mathbf{t}_{\tilde Z}})=\sqrt{\overline{\left|\tilde{\mathbf{b}}_{p(i)}+\tilde{\mathbf{t}}_i-\tilde{\mathbf{a}}_i-\overline{\tilde{\mathbf{b}}_j}-\overline{\tilde{\mathbf{t}}_j}+\overline{\tilde{\mathbf{a}}_j}\right|^2}}
	,\end{equation}
	where $\tilde{\mathbf{b}}_i=V\Sigma^{-\frac{1}{2}}U^\text{T}\mathbf{b}_i$ is deformed from the $L_B$-motif that defines the PCT, and so is $\tilde{\mathbf{t}}_i=V\Sigma^{-\frac{1}{2}}U^\text{T}\mathbf{t}_i$.
\end{lem}
\begin{proof}[\indent Proof]
	For any $\mathbf{v}_1,\cdots,\mathbf{v}_{\tilde Z}\in\mathbb{R}^3$ and $\boldsymbol{\tau}\in\mathbb{R}^3$, we have
	\begin{align}
		&\phantom{{}={}}\overline{\lvert\mathbf{v}_i+\boldsymbol{\tau}\rvert^2}
	      \\&=\overline{\lvert(\mathbf{v}_i-\overline{\mathbf{v}_j})+(\boldsymbol{\tau}+\overline{\mathbf{v}_j})\rvert^2}
	      \\&=\overline{\lvert \mathbf{v}_i-\overline{\mathbf{v}_j}\rvert^2}+2\,\overline{(\mathbf{v}_i-\overline{\mathbf{v}_j})\cdot(\boldsymbol{\tau}+\overline{\mathbf{v}_j})}+\overline{\lvert \boldsymbol{\tau}+\overline{\mathbf{v}_j}\rvert^2 }
	      \\&=\overline{\lvert \mathbf{v}_i-\overline{\mathbf{v}_j}\rvert^2}+2\,\overline{(\mathbf{v}_i-\overline{\mathbf{v}_j})}\cdot(\boldsymbol{\tau}+\overline{\mathbf{v}_j})+\lvert \boldsymbol{\tau}+\overline{\mathbf{v}_j}\rvert^2
	      \\&=\overline{\lvert \mathbf{v}_i-\overline{\mathbf{v}_j}\rvert^2}+\lvert \boldsymbol{\tau}+\overline{\mathbf{v}_j}\rvert^2
	.\end{align}
	We can see that $\overline{\lvert \mathbf{v}_i+\boldsymbol{\tau}\rvert^2}$ attains its minimum $\overline{\lvert \mathbf{v}_i-\overline{\mathbf{v}_j}\rvert^2 }$ if and only if $\boldsymbol{\tau}=-\overline{\mathbf{v}_i}$.
	Hence, when $\ell=2$, we have
	\begin{equation}
		d(\mathcal{J})=\sqrt{\overline{\left|\tilde{\mathcal{J}}(\tilde{\mathbf{a}}_i)-\tilde{\mathbf{a}}_i-\overline{\tilde{\mathcal{J}}(\tilde{\mathbf{a}}_j)}+\overline{\tilde{\mathbf{a}}_j}\right|^2}}
	,\end{equation}
	which is identical to Eq.~(\ref{djbpi}).
\end{proof}
\begin{lem}\label{lem:finite-z}
	Let $\mathcal{A}$ be a crystal structure and $L$ a sublattice of $\mathcal{A}$.
	If we additionally require that $\mathcal{A}\subset\mathbb{R}^3$ has no cluster points, then $\lvert \mathcal{A}/L\rvert$ is finite.
\end{lem}
\begin{proof}[\indent Proof]
	We assume by contradiction that the quotient set $\mathcal{A}/L$ is infinite.
	Take $\mathbf{t}_1,\mathbf{t}_2,\mathbf{t}_3\in L$ that satisfy Eq.~(\ref{t1tnl}) and consider the parallelepiped
	\begin{equation}
		P=\left\{\sum_{i=1}^3x_i\mathbf{t}_i\,\middle|\,x_1,x_2,x_3\in[0,1)\right\}
	.\end{equation}
	Any element in $\mathbb{R}^3$ can be uniquely decomposed as
	\begin{equation}
	  \mathbf{v}=\sum_{i=1}^3c_i\mathbf{t}_i
	,\end{equation}
	where $c_1,c_2,c_3\in\mathbb{R}$.
	Therefore, the mapping
	\begin{align}
		\pi\colon\qquad \mathbb{R}^3&\to P,\\
		\sum_{i=1}^3c_i\mathbf{t}_i&\mapsto\sum_{i=1}^3\big(c_i-\lfloor c_i\rfloor\big)\mathbf{t}_i
	,\end{align}
	is well-defined.
	Note that the image of each $\mathbf{a}+L\in\mathcal{A}/L$ under $\pi$ is a single element in $P\cap\mathcal{A}$, and
	\begin{equation}
		\pi(\mathbf{a}_1+L)=\pi(\mathbf{a}_2+L)\quad\implies\quad \mathbf{a}_1+L=\mathbf{a}_2+L
	,\end{equation}
	i.e., $\pi\colon\mathcal{A}/L\to P\cap\mathcal{A}$ is injective.
	Since $\mathcal{A}/L$ is infinite, its image $\pi(\mathcal{A}/L)$ is also infinite, and so is $P\cap\mathcal{A}$.
	The Bolzano-Weierstrass theorem then states that $P\cap\mathcal{A}$ must have a cluster point.
\end{proof}
\begin{lem}\label{lem:intersect}
	Let $L$ be a lattice, and $L_1,L_2$ its sublattice.
	The intersection set $L_1\cap L_2$ is full rank.
\end{lem}
\begin{proof}[\indent Proof]
	Let $C,C_1,C_2$ be base matrices of $L,L_1,L_2$, respectively.
	We only need to prove that there exist nonsingular $N_1,N_2\in\mathbb{Z}^{3\times 3}$ such that
	\begin{equation}\label{c1n1c}
	  C_1N_1=C_2N_2
	.\end{equation}
	Lemma~\ref{lem:sublat} states that $C_1=CM_1$ and $C_2=CM_2$ for some nonsingular $M_1,M_2\in\mathbb{Z}^{3\times 3}$.
	Therefore, Eq.~\ref{c1n1c} is equivalent to
	\begin{equation}\label{n1m1m}
	  N_1=M_1^{-1}M_2N_2
	.\end{equation}
	Since $M_1^{-1}M_2$ is rational, we can always let $m$ be a common multiple of the denominators of its elements.
	Then we have nonsingular $N_1=mM_1^{-1}M_2$ and $N_2=mI$ that satisfies Eq.~\ref{n1m1m}.
\end{proof}
\section{Details of Algorithm~\ref{alg:imt}}\label{append:alg}
For empirical reasons, we generate $S_0\in\Omega_{S_0}$ using the SVD $S_0=U\Sigma V^\text{T}$.
Here, $U$ and $V$ are independently, uniformly sampled in $\operatorname{SO}(3)$, while the generation of $\Sigma=\operatorname{diag}(s_1,s_2,s_3)$ depends on $U$ and $V$.
We sample $(s_1,s_2,s_3)$ within the region $w(U\Sigma V^\text{T})\le w_\text{max}$ according to the distribution
\begin{equation}
	f(s_1,s_2,s_3)\propto \left\lvert s_1^2-s_2^2\right\rvert\left\lvert s_2^2-s_3^2\right\rvert\left\lvert s_1^2-s_3^2\right\rvert
,\end{equation}
and then sort them to restore $s_1\ge s_2\ge s_3$.
This procedure ensures that the projection of $S_0$ on $\operatorname{SO}(3)$, i.e., $UV^\text{T}$, is uniformly distributed, while for each $(U,V)$, the conditional distribution of $S_0$ follows the standard measure on $\mathbb{R}^{3\times 3}$.
The following algorithm analysis is based on this distribution, but one may well use other distributions to improve the performance of Algorithm~\ref{alg:imt}.
We begin by considering the probability that a given IMT $(H_A,H_B,Q)$ is generated in a single iteration.
This can be estimated using the intersection between the cyan region and the pink hypercube centered at $Q$ in Fig.~\ref{fig:alg1}.
Given that strains in SSPTs are usually small, the intersection can be considered as a thin vicinity of the 3D section of Eq.~(\ref{round}) intercepted by $\operatorname{SO}(3)$.
Therefore, the probability of $(H_A,H_B,Q)$ being generated is approximately equal to the probability that $\phi_{H_A,H_B}(UV^\text{T})$ is generated within this section.
Since the cyan region contains $Q$, the volume of this section can be assumed to be no less than 1.
On the other hand, $UV^\text{T}$ is uniformly distributed over $\operatorname{SO}(3)$, whose volume as a 3D submanifold is $16\sqrt{2}\,\pi^2$.
Hence, after the deformation by $\phi_{H_A,H_B}$, the probability density near $Q$ is at least
\begin{equation}
	\eta_{H_A,H_B}=\frac{\sigma_7(\phi_{H_A,H_B})\sigma_8(\phi_{H_A,H_B})\sigma_9(\phi_{H_A,H_B})}{16 \sqrt{2}\, \pi^2}
,\end{equation}
where
\begin{equation}
	\phi_{H_A,H_B}=(H_B^{-1}C_B^{-1})\otimes(H_A^\text{T}C_A^\text{T})
\end{equation}
is treated as a $9\times 9$ matrix.
Now we consider the probability that no IMT in a given congruence class is generated in a single iteration.
Since each IMT corresponds to an SLM, $(H_A',H_B',Q')$ is congruent to $(H_A,H_B,Q)$ if and only if there exist $R_A\in G_A'$ and $R_B\in G_B'$ such that
\begin{align}
	H_A'&=\operatorname{hnf}(\hat{R}_AH_A),\\
	H_B'&=\operatorname{hnf}(\hat{R}_BH_B),\\
	H_B'Q'H_A'^{-1}&=H_BQH_A^{-1}
,\end{align}
where $\hat{R}_A=C_A^{-1}R_AC_A$ and $\hat{R}_B=C_B^{-1}R_BC_B$ are unimodular matrices that describe how $R_A$ and $R_B$ change the base matrices of $L_A$ and $L_B$, respectively.
Therefore, the probability that no IMT congruent to $(H_A,H_B,Q)$ is generated has a lower bound
\begin{equation}
	\zeta_{H_A,H_B}=\prod_{R_A\in G_A'}\prod_{R_B\in G_B'} (1-\eta_{H_A',H_B'})
.\end{equation}
Replacing $(1-\eta)$ in Eq.~(\ref{ithln}) with this expression yields a significantly smaller
\begin{equation}
	i_\text{th}=\frac{\ln \epsilon}{\ln \displaystyle\max_{H_A,H_B}\zeta_{H_A,H_B}}
,\end{equation}
which can be computed numerically before Algorithm~\ref{alg:imt} starts.
%
%\bibliography{ref}

%

\end{document}